\documentclass[12pt,cls,onecolumn]{IEEEtran}
\usepackage{graphicx,amsmath,amssymb,epsfig, amsfonts, cite, latexsym, cuted, multicol, multirow, subfigure, stfloats, array, tabularx}
\usepackage{subeqnarray}
\usepackage{color}
\usepackage{setspace}
\usepackage{anysize}

\begin{document}

\title{Opportunistic Downlink Interference Alignment}
\author{\large Hyun Jong Yang, \emph{Member}, \emph{IEEE}, Won-Yong Shin, \emph{Member}, \emph{IEEE}, \\ Bang Chul Jung, \emph{Senior Member}, \emph{IEEE}, Changho Suh, \emph{Member}, \emph{IEEE}, and
Arogyaswami Paulraj, \emph{Fellow}, \emph{IEEE} \\
\thanks{H. J. Yang is with the School of Electrical and Computer Engineering, UNIST, Ulsan 689-798, Republic of Korea (E-mail:
hjyang@unist.ac.kr).}
\thanks{W.-Y. Shin is with the Department of Computer Science and
Engineering, Dankook University, Yongin 448-701, Republic of Korea
(E-mail: wyshin@dankook.ac.kr).}
\thanks{B. C. Jung (corresponding author) is with the Department of Information and Communication Engineering, Gyeongsang National
University, Tongyeong 650-160, Republic of Korea (E-mail:
bcjung@gnu.ac.kr).}
\thanks{C. Suh is with the Department of Electrical Engineering, KAIST, Daejeon 305-701, Republic of Korea (E-mail: chsuh@kaist.ac.kr).}
\thanks{A. Paulraj is with the Department of Electrical Engineering,
Stanford University, Stanford, CA 94305 (email:
apaulraj@stanford.edu). }
} \maketitle


\markboth{Submitted to IEEE Transactions on Signal Processing} {Yang
{\em et al.}: Opportunistic Downlink Interference Alignment}


\newtheorem{definition}{Definition}
\newtheorem{theorem}{Theorem}
\newtheorem{lemma}{Lemma}
\newtheorem{example}{Example}
\newtheorem{corollary}{Corollary}
\newtheorem{proposition}{Proposition}
\newtheorem{conjecture}{Conjecture}
\newtheorem{remark}{Remark}

\def \diag{\operatornamewithlimits{diag}}
\def \min{\operatornamewithlimits{min}}
\def \max{\operatornamewithlimits{max}}
\def \log{\operatorname{log}}
\def \max{\operatorname{max}}
\def \rank{\operatorname{rank}}
\def \out{\operatorname{out}}
\def \exp{\operatorname{exp}}
\def \arg{\operatorname{arg}}
\def \E{\operatorname{E}}
\def \tr{\operatorname{tr}}
\def \SNR{\operatorname{SNR}}
\def \dB{\operatorname{dB}}
\def \ln{\operatorname{ln}}

\def \be {\begin{eqnarray}}
\def \ee {\end{eqnarray}}
\def \ben {\begin{eqnarray*}}
\def \een {\end{eqnarray*}}

\begin{abstract}
In this paper, we propose an opportunistic downlink interference
alignment (ODIA) for interference-limited cellular downlink, which
intelligently combines user scheduling and downlink IA techniques.
The proposed ODIA not only efficiently reduces the effect of
inter-cell interference from other-cell base stations (BSs) but also
eliminates intra-cell interference among spatial streams in the same
cell. We show that the minimum number of users required to achieve a
target degrees-of-freedom (DoF) can be fundamentally reduced, i.e.,
the fundamental user scaling law can be improved by using the ODIA,
compared with the existing downlink IA schemes. In addition, we
adopt a limited feedback strategy in the ODIA framework, and then
analyze the required number of feedback bits leading to the same
performance as that of the ODIA assuming perfect feedback. We also
modify the original ODIA in order to further improve sum-rate, which
achieves the optimal multiuser diversity gain, i.e., $\log \log N$,
per spatial stream even in the presence of downlink inter-cell
interference, where $N$ denotes the number of users in a cell.
Simulation results show that the ODIA significantly outperforms
existing interference management techniques in terms of sum-rate in
realistic cellular environments. Note that the ODIA operates in a
distributed and decoupled manner, while requiring no information
exchange among BSs and no iterative beamformer optimization between
BSs and users, thus leading to an easier implementation.
\end{abstract}


\newpage


\section{Introduction}
Interference management has been taken into account as one of the
most challenging issues to increase the throughput of cellular
networks serving multiple users. In multiuser cellular environments,
each receiver may suffer from intra-cell and inter-cell
interference.
Interference alignment (IA) was proposed by fundamentally solving
the interference problem when there are multiple communication
pairs~\cite{V_Cadambe08_TIT}. It was shown that the IA scheme can
achieve the optimal degrees-of-freedom (DoF)\footnote{It is referred
that `optimal' DoF is achievable if the outer-bound on DoF for given
network configuration is achievable. } in the multiuser interference
channel with time-varying channel coefficients. Subsequent studies
have shown that the IA is also useful and indeed achieves the
optimal DoF in various wireless multiuser network setups:
multiple-input multiple-output (MIMO) interference
channels~\cite{K_Gomadam11_TIT, T_Gou10_TIT} and cellular
networks~\cite{C_Suh11_TC,C_Suh08_Allerton}. In particular, IA
techniques~\cite{C_Suh11_TC,C_Suh08_Allerton} for cellular uplink
and downlink networks, also known as the interfering multiple-access
channel (IMAC) or interfering broadcast channel (IBC), respectively,
have received much attention.
 The existing IA framework for cellular networks, however, still has
several practical challenges: the scheme proposed
in~\cite{C_Suh08_Allerton} requires arbitrarily large
frequency/time-domain dimension extension, and the scheme proposed
in~\cite{C_Suh11_TC} is based on iterative optimization of
processing matrices and cannot be optimally extended to an arbitrary
downlink cellular network in terms of achievable DoF.
%

In the literature, there are some results on the usefulness of
fading in single-cell downlink broadcast channels, where one can
obtain multiuser diversity gain along with user scheduling as the
number of users is sufficiently large: opportunistic
scheduling~\cite{R_Knopp95_ICC}, opportunistic
beamforming~\cite{P_Viswanath02_TIT}, and random
beamforming~\cite{M_Sharif05_TIT}. Scenarios exploiting multiuser
diversity gain have been studied also in ad hoc
networks~\cite{W_Shin14_TIT}, cognitive radio
networks~\cite{T_Ban09_TWC}, and cellular
networks~\cite{W_Shin12_TC}.

Recently, the concept of opportunistic IA~(OIA) was introduced
in~\cite{B_Jung11_CL,B_Jung11_TC,H_Yang13_TWC} for the $K$-cell
uplink network (i,e., IMAC model), where there are one $M$-antenna
base station (BS) and $N$ users in each cell. The OIA scheme
incorporates user scheduling into the classical IA framework by
opportunistically selecting $S$ ($S\le M$) users amongst the $N$
users in each cell in the sense that inter-cell interference is
aligned at a pre-defined interference space. It was shown
in~\cite{B_Jung11_TC,H_Yang13_TWC} that one can asymptotically
achieve the optimal DoF if the number of users in a cell is beyond a
certain value, i.e., if a certain user scaling condition is
guaranteed. For the $K$-cell downlink network (i.e., IBC model)
assuming one $M$-antenna base station (BS) and $N$ per-cell users,
studies on the OIA have been conducted in~\cite{W_Shin12_IEICE,
J_Jose12_Allerton, J_Lee13_TWC, H_Nguyen13_TSP,H_Nguyen13_arXiv,
J_Lee13_arXiv}. More specifically, the user scaling condition for
obtaining the optimal DoF was characterized for the $K$-cell
multiple-input single-output (MISO) IBC~\cite{W_Shin12_IEICE}, and
then such an analysis of the DoF achievability was extended to the
$K$-cell MIMO IBC with $L$ receive antennas at each
user~\cite{J_Jose12_Allerton, J_Lee13_TWC,
H_Nguyen13_TSP,H_Nguyen13_arXiv, J_Lee13_arXiv}---full DoF can be
achieved asymptotically, provided that $N$ scales faster than
${\mathsf{SNR}}^{KM-L}$, for the $K$-cell MIMO IBC using
OIA~\cite{H_Nguyen13_arXiv, J_Lee13_arXiv}, where SNR denotes the
received signal-to-noise ratio.

In this paper, we propose an \textit{opportunistic downlink IA
(ODIA)} framework as a promising interference management technique
for $K$-cell downlink networks, where each cell consists of one BS
with $M$ antennas and $N$ users having $L$ antennas each.
The proposed ODIA jointly takes into account user scheduling and
downlink IA issues. In particular, inspired by the precoder design
in~\cite{C_Suh11_TC}, we use two cascaded beamforming matrices to
construct our precoder at each BS. To design the first transmit
beamforming matrix, we use a user-specific beamforming, which
conducts a linear zero-forcing (ZF) filtering and thus eliminates
intra-cell interference among spatial streams in the same cell. To
design the second transmit beamforming matrix, we use a
predetermined reference beamforming matrix, which plays the same
role of random beamforming for cellular
downlink~\cite{W_Shin12_IEICE, H_Nguyen13_arXiv, J_Lee13_arXiv} and
thus efficiently reduces the effect of inter-cell interference from
other-cell BSs. On the other hand, the receive beamforming vector is
designed at each user in the sense of minimizing the total amount of
received inter-cell interference using \textit{local} channel state
information (CSI) in a decentralized manner. Each user feeds back
both the effective channel vector and the quantity of received
inter-cell interference to its home-cell BS. The user selection and
transmit beamforming at the BSs and the design of receive
beamforming at the users are completely decoupled. Hence, the ODIA
operates in a fully distributed manner while requiring no
information exchange among BSs and no iterative optimization between
transmitters and receivers, thereby resulting in an easier
implementation.


The main contribution of this paper is four-fold as follows.
\begin{itemize}
\item  We first show that the minimum number of users
required to achieve $S$ DoF ($S\le M$) can be fundamentally reduced
to $\mathsf{SNR}^{(K-1)S-L+1}$ by using the ODIA at the expense of
acquiring perfect CSI at the BSs from users, compared to the
existing downlink IA schemes requiring the user scaling law
$N=\omega(\mathsf{SNR}^{KS-L})$~\cite{H_Nguyen13_arXiv,
J_Lee13_arXiv},\footnote{$f(x) = \omega(g(x))$ implies that $\lim_{x
\rightarrow \infty} \frac{g(x)}{f(x)}=0$.} where $S$ denotes the
number of spatial streams per cell. The interference decaying rate
with respect to $N$ for given SNR is also characterized in regards
to the derived user scaling law.
\item We introduce a
limited feedback strategy in the ODIA framework, and then analyze
the required number of feedback bits leading to the same DoF
performance as that of the ODIA assuming perfect feedback, which is
given by $\omega\left( \log_2 \mathsf{SNR}\right)$.
\item  We modify the user scheduling part of the ODIA to achieve optimal multiuser diversity gain, i.e., $\log\log N$ per stream
even in the presence of downlink inter-cell interference.
\item To verify the ODIA schemes, we perform numerical
evaluation via computer simulations. Simulation results show that
the proposed ODIA significantly outperforms existing interference
management and user scheduling techniques in terms of sum-rate in
realistic cellular environments.
\end{itemize}

The remainder of this paper is organized as follows. Section
\ref{SEC:system} describes the system and channel models. Section
\ref{SEC:OIA} presents the overall procedure of the proposed ODIA.
In Section \ref{sec:achievability}, the DoF achievablility result is
shown. Section \ref{SEC:OIA_limited} presents the ODIA scheme with
limited feedback. In Section \ref{SEC:Threhold_ODIA}, the
achievability of the spectrally efficient ODIA leading to a better
sum-rate performance is characterized. Numerical results are shown
in Section \ref{SEC:Sim}. Section \ref{SEC:Conc} summarizes the
paper with some concluding remarks.


\section{System and Channel Models} \label{SEC:system}
We consider a $K$-cell MIMO IBC where each cell consists of a BS
with $M$ antennas and $N$ users with $L$ antennas each. The number
of selected users in each cell is denoted by $S (\le M)$. It is
assumed that each selected user receives a single spatial stream. To
consider nontrivial cases, we assume that $L < (K-1)S +1$, because
all inter-cell interference can be completely canceled at the
receivers (i.e., users) otherwise. The channel matrix from the
$k$-th BS to the $j$-th user in the $i$-th cell is denoted by
$\mathbf{H}_{k}^{[i,j]}\in \mathbb{C}^{L \times M}$, where $i,k\in
\mathcal{K} \triangleq \{ 1, \ldots, K\}$ and $j \in \mathcal{N}
\triangleq \{1, \ldots, N\}$. Each element of $\mathbf{H}_k^{[i,j]}$
is assumed to be independent and identically distributed (i.i.d.)
according to $\mathcal{CN}(0,1)$. In addition, quasi-static
frequency-flat fading is assumed, i.e., channel coefficients are
constant during one transmission block and change to new independent
values for every transmission block. Owing to the channel
reciprocity of time-division duplexing (TDD) systems, the $j$-th
user in the $i$-th cell can estimate the channels
$\mathbf{H}_{k}^{[i,j]}$, $k=1, \ldots, K$, using pilot signals sent
from all the BSs, i.e., the local CSI at the transmitters is
available. Figure \ref{fig:system_model} shows an example of the
MIMO IBC model, where $K=3$, $M=3$, $S=2$, $L=3$, and $N=2$. The
details in the figure will be described in the subsequent section.

\begin{figure} \label{fig:system_model}
\begin{center}
  \includegraphics[width=0.80\textwidth]{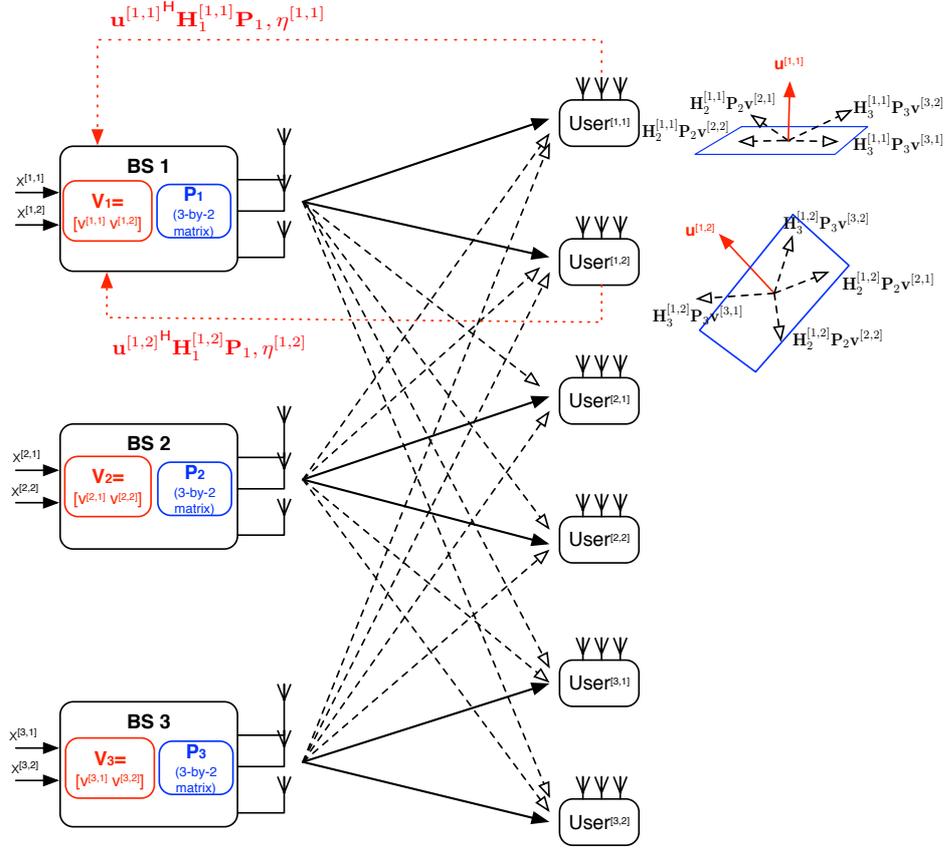}\\
  \caption{The MIMO IBC model, where $K=3$, $M=3$, $S=2$, $L=3$, and $N=2$.}\label{fig:system_model}
  \end{center}
\end{figure}


\section{Proposed ODIA} \label{SEC:OIA}

 We first describe the overall procedure of our proposed ODIA scheme for the MIMO IBC, and then define its achievable sum-rate and DoF.
\subsection{Overall Procedure} \label{subsec:overall}

The ODIA scheme is described according to the following four steps.
\subsubsection{Initialization (Broadcast of Reference Beamforming Matrices)}
First, as illustrated in Fig.~\ref{fig:system_model}, the precoding
matrix at each BS is composed of the product of a predetermined
reference beamforming matrix, denoted by $\mathbf{P}_k$, and a
user-specific beamforming matrix, denoted by $\mathbf{V}_k$. In this
step, we mainly focus on the design of $\mathbf{P}_k$. Specifically,
the reference beamforming matrix at the BS in the $k$-th cell is
given by $\mathbf{P}_k = \left[ \mathbf{p}_{1,k}, \ldots,
\mathbf{p}_{S,k}\right]$, where $\mathbf{p}_{s,k} \in \mathbb{C}^{M
\times 1}$ is an orthonormal basis for $k\in \mathcal{K}$ and $s =1,
\ldots, S$.
 Each BS independently generates $\mathbf{p}_{k,s}$ according to the isotropic distribution over the $M$-dimensional unit sphere.
 If the reference beamforming matrix is generated in a pseudo-random fashion, BSs do not need to broadcast them to users.
 Then, the $j$-th user in the $i$-th cell obtains $\mathbf{H}^{[i,j]}_{k}$ and $\mathbf{P}_k$, $k=1, \ldots, K$.
%

\subsubsection{Receive Beamforming  \& Scheduling Metric Feedback}
In the second step, we explain how to decide a user scheduling
metric at each user along with given receive beamforming, where the
design of receive beamforming will be explained in
Section~\ref{sec:achievability}. Let $\mathbf{u}^{[i,j]} \in
\mathbb{C}^{L \times 1}$ denote the unit-norm weight vector at the
$j$-th user in the $i$-th cell, i.e., $\left\| \mathbf{u}^{[i,j]}
\right\|^2 = 1$. Since the user-specific beamforming $\mathbf{V}_k$
will be utilized only to cancel intra-cell interference out,
$\mathbf{V}_k$ does not change the inter-cell interference level at
each user, which will be specified later. Thus, from the notion of
$\mathbf{P}_k$ and $\mathbf{H}^{[i,j]}_{k}$, the $j$-th user in the
$i$-th cell can compute the quantity of received interference from
the $k$-th BS while using its receive beamforming vector
$\mathbf{u}^{[i,j]}$, which is given by
\begin{align}\label{eq:eta_tilde}
\tilde{\eta}^{[i,j]}_{k} &= \left\|
{\mathbf{u}^{[i,j]}}^{\mathsf{H}}\mathbf{H}_{k}^{[i,j]} \mathbf{P}_k
\right\|^2,
\end{align}
where $i\in \mathcal{K}$, $j \in \mathcal{N} $, and $k\in
\mathcal{K}\setminus i= \{1, \ldots, i-1, i+1, \ldots, K\}$. Using
(\ref{eq:eta_tilde}), the scheduling metric at the $j$-th user in
the $i$-th cell, denoted by $\eta^{[i,j]}$, is defined as the sum of
received interference power from other cells. That is,
\begin{align} \label{eq:eta}
\eta^{[i,j]} &= \sum_{k=1, k\neq i}^{K} \tilde{\eta}^{[i,j]}_{k}.
\end{align}

As illustrated in Fig. \ref{fig:system_model}, each user feeds the
metric in (\ref{eq:eta}) back to its home-cell BS.
In addition to the scheduling metric in (\ref{eq:eta}), each user
needs to feed its effective channel vector back, so that the
user-specific beamforming $\mathbf{V}_k$ is designed at each BS. The
effective channel vector of the $j$-th user in the $i$-th cell is
given by
\begin{equation} \label{eq:effective_CH}
\mathbf{f}_{i}^{[i,j]} \triangleq
\left({\mathbf{u}^{[i,j]}}^{\mathsf{H}} \mathbf{H}^{[i,j]}_i
\mathbf{P}_ i\right)^{\mathsf{H}}.
\end{equation}

\subsubsection{User Scheduling}
Upon receiving $N$ users' scheduling metrics in the serving cell,
each BS selects $S$ users having the metrics up to the $S$-th
smallest one. Without loss of generality, the indices of selected
users in every cell are assumed to be $(1, \ldots, S)$.
In this and subsequent sections, we focus on how to simply design a
user scheduling method to guarantee the optimal DoF. An enhanced
scheduling algorithm jointly taking into account the effective
channel in (\ref{eq:effective_CH}) and the received interference
level in (\ref{eq:eta}) may provide a better performance in terms of
sum-rate, which shall be discussed in Section
\ref{SEC:Threhold_ODIA}.


\subsubsection{Transmit Beamforming \& Downlink Data Transmission} The signal vector at the $i$-th BS transmitted to the $j$-th user in the $i$-th cell is given by $\mathbf{v}^{[i,j]}x^{[i,j]}$, where $x^{[i,j]}$ is the transmit symbol with power of $1/S$, and the user-specific beamforming matrix for $S$ users is given by $\mathbf{V}_i = \left[ \mathbf{v}^{[i,1]}, \ldots, \mathbf{v}^{[i,S]}\right]$, where $\mathbf{v}^{[i,s]} \in \mathbb{C}^{S \times 1}$, $i\in \mathcal{K}$.
Denoting the transmit symbol vector of the $i$-th cell by
$\mathbf{x}_i = \left[ x^{[i,1]}, \ldots, x^{[i,S]}\right]^T$, the
received signal vector at the $j$-th user in the $i$-th cell is then
written as
\begin{align} \label{eq:rec_vector}
\mathbf{y}^{[i,j]} &= \mathbf{H}_i^{[i,j]}\mathbf{P}_i \mathbf{V}_i \mathbf{x}_i + \sum_{k=1, k\neq i}^{K} \mathbf{H}_k^{[i,j]}\mathbf{P}_k \mathbf{V}_k \mathbf{x}_k + \mathbf{z}^{[i,j]}  \nonumber \\
&= \underbrace{\mathbf{H}_i^{[i,j]}\mathbf{P}_i \mathbf{v}^{[i,j]} x^{[i,j]}}_{\textsf{desired signal}} +  \underbrace{\sum_{s=1, s\neq j}^{S} \mathbf{H}_i^{[i,j]}\mathbf{P}_i \mathbf{v}^{[i,s]} x^{[i,s]}}_{\textsf{intra-cell interference}} \nonumber \\
& \hspace{20pt}+ \underbrace{\sum_{k=1, k\neq i}^{K}
\mathbf{H}_k^{[i,j]}\mathbf{P}_k \mathbf{V}_k
\mathbf{x}_k}_{\textsf{inter-cell interference}} +
\mathbf{z}^{[i,j]},
\end{align}
  where $\mathbf{z}^{[i,j]} \in \mathbb{C}^{L \times 1}$ denotes the additive white Gaussian noise vector, each element of which is i.i.d. complex Gaussian with zero mean and the variance of $\mathsf{{SNR}^{-1}}$.
  The received signal vector at the $j$-th user in the $i$-th cell after receive beamforming, denoted by $\tilde{y}^{[i,j]} = {\mathbf{u}^{[i,j]}}^{\mathsf{H}} \mathbf{y}^{[i,j]}$, can be rewritten as:
\begin{align}\label{eq:rec_vector_after_BF}
\tilde{y}^{[i,j]}
&= {\mathbf{f}_{i}^{[i,j]}}^{\mathsf{H}} \mathbf{v}^{[i,j]}
x^{[i,j]}  +{\mathbf{f}_{i}^{[i,j]}}^{\mathsf{H}}\sum_{s=1, s\neq
j}^{S}
\mathbf{v}^{[i,s]} x^{[i,s]}\nonumber \\
&\hspace{10pt}  + \sum_{k=1, k\neq i}^{K}
{\mathbf{f}_{k}^{[i,j]}}^{\mathsf{H}} \mathbf{V}_k \mathbf{x}_k +
{\mathbf{u}^{[i,j]}}^{\mathsf{H}}\mathbf{z}^{[i,j]},
\end{align}
where
${\mathbf{f}_{k}^{[i,j]}}^{\mathsf{H}}={\mathbf{u}^{[i,j]}}^{\mathsf{H}}\mathbf{H}_k^{[i,j]}\mathbf{P}_k$.
By selecting users with small $\eta^{[i,j]}$ in (\ref{eq:eta}),
$\mathbf{H}_k^{[i,j]}\mathbf{P}_k$ tends to be orthogonal to the
receive beamforming vector $\mathbf{u}^{[i,j]}$; thus, inter-cell
interference channel matrices
$\mathbf{H}_k^{[i,j]}\mathbf{P}_k\mathbf{V}_k$ in
(\ref{eq:rec_vector_after_BF}) also tend to be orthogonal to
$\mathbf{u}^{[i,j]}$ as illustrated in Fig. \ref{fig:system_model}.

To cancel out intra-cell interference, the user-specific beamforming
matrix $\mathbf{V}_i \in \mathbb{C}^{S \times S}$is given by
\begin{align} \label{eq:ZF_BF}
\mathbf{V}_i &= [\mathbf{v}^{[i,1]},\mathbf{v}^{[i,2]}, \ldots, \mathbf{v}^{[i,S]}] \nonumber \\
 &= \begin{bmatrix}
       {\mathbf{u}^{[i,1]}}^{\mathsf{H}} \mathbf{H}_i^{[i,1]} \mathbf{P}_i  \\
       {\mathbf{u}^{[i,2]}}^{\mathsf{H}} \mathbf{H}_i^{[i,2]} \mathbf{P}_i  \\
       \vdots \\
       {\mathbf{u}^{[i,S]}}^{\mathsf{H}} \mathbf{H}_i^{[i,S]} \mathbf{P}_i
     \end{bmatrix}^{-1} \cdot  \begin{bmatrix}
       \sqrt{\gamma^{[i,1]}} & 0 & \cdots & 0  \\
       0 & \sqrt{\gamma^{[i,2]}} & \cdots & 0  \\
       \vdots & \vdots & \ddots & \vdots \\
       0 & 0 & \cdots & \sqrt{\gamma^{[i,S]}}  \\
     \end{bmatrix},
\end{align}
where $\sqrt{\gamma^{[i,j]}}$ denotes a normalization factor for
satisfying the unit-transmit power constraint. In consequence, the
received signal can be simplified to
\begin{align}\label{eq:rec_vector_ZF_BF}
\tilde{y}^{[i,j]} &= \sqrt{\gamma^{[i,j]}} x^{[i,j]} \nonumber  \\
& \hspace{20pt}+ \underbrace{\sum_{k=1, k\neq i}^{K}
{\mathbf{f}_{k}^{[i,j]}}^{\mathsf{H}}\mathbf{V}_k
\mathbf{x}_k}_{\textsf{inter-cell interference}} +
{\mathbf{u}^{[i,j]}}^{\mathsf{H}}\mathbf{z}^{[i,j]},
\end{align}
which thus does not contain the intra-cell interference term.

As in
\cite{N_Jindal06_TIT,T_Yoo07_JSAC,J_Thukral09_ISIT,R_Krishnamachari10_ISIT,S_Pereira07_Asilomar,B_Jung11_TC},
we assume no loss in exchanging signaling messages such as
information of effective channels, scheduling metrics, and receive
beamforming vectors.

\subsection{Achievable Sum-Rate and DoF}\label{subsec:sum_rate}
From (\ref{eq:rec_vector_ZF_BF}), the achievable rate of the $j$-th
user in the $i$-th cell is given by
\begin{align} \label{eq:data_rate_single_user}
R^{[i,j]}
&= \log_2 \left( 1+ \frac{ \gamma^{[i,j]} \cdot
|x^{[i,j]}|^2}{\left|{{\mathbf{u}^{[i,j]}}^{\mathsf{H}}}
\mathbf{z}^{[i,j]}\right|^2+\tilde{I}^{[i,j]}} \right) \nonumber \\
& =\log_2 \left( 1+ \frac{ \gamma^{[i,j]} }{\frac{S}{\mathsf{SNR}} +
\sum_{k=1, k\neq i}^{K} \sum_{s=1}^{S} \left|
{\mathbf{f}_{k}^{[i,j]}}^{\mathsf{H}} \mathbf{v}^{[k,s]}\right|^2 }
\right),
\end{align}
where $\tilde{I}^{[i,j]} \triangleq \sum_{k=1, k\neq i}^{K}
\left|{\mathbf{f}_{k}^{[i,j]}}^{\mathsf{H}} \mathbf{V}_k
\mathbf{x}_k\right|^2$. Using (\ref{eq:data_rate_single_user}), the
achievable total DoF can be defined as
\begin{equation}
\textrm{DoF} = \lim_{\textsf{SNR} \rightarrow \infty}
\frac{\sum_{i=1}^{K}\sum_{j=1}^{S}R^{[i,j]}}{\log \textsf{SNR}}.
\end{equation}

\section{DoF Achievability}\label{sec:achievability}
In this section, we characterize the DoF achievability in terms of
the user scaling law with the optimal receive beamforming technique.
To this end, we start with the receive beamforming design that
maximizes the achievable DoF. For given channel instance, from
(\ref{eq:data_rate_single_user}), each user can attain the maximum
DoF of 1 if and only if the interference $\sum_{k=1, k\neq i}^{K}
\sum_{s=1}^{S}
\Big|{\mathbf{f}_{k}^{[i,j]}}^{\mathsf{H}}\mathbf{v}^{[k,s]}\Big|^2
\cdot \mathsf{SNR}$ remains constant for increasing SNR. Note that
$R^{[i,j]}$ can be bounded as
\begin{align}
&R^{[i,j]}\ge \log_2 \left( 1+ \frac{ \gamma^{[i,j]} }{ \frac{S}{\mathsf{SNR}}+ \sum_{k=1, k\neq i}^{K} \sum_{s=1}^{S} \left\| \mathbf{f}_{k}^{[i,j]}\right\|^2 \left\| \mathbf{v}^{[k,s]}\right\|^2 } \right) \label{eq:data_rate_single_user_bound} \\
& \ge \log_2 \left( 1+ \frac{ \gamma^{[i,j]} }{ \frac{S}{\mathsf{SNR}}+ \sum_{k\neq i}^{K} \sum_{s=1}^{S} \left\| \mathbf{f}_{k}^{[i,j]} \right\|^2 \left\| \mathbf{v}^{(\max)}_{i}\right\|^2 } \right) \label{eq:data_rate_single_user_bound2}\\
& = \log_2\left( \mathsf{SNR}\right) + \log_2
\left(\frac{1}{\mathsf{SNR}}+ \frac{ \frac{\gamma^{[i,j]}}{\left\|
\mathbf{v}^{(\max)}_{i}\right\|^2}}{ \frac{S}{\left\|
\mathbf{v}^{(\max)}_{i}\right\|^2}+ I^{[i,j]} } \right)
\label{eq:data_rate_single_user_bound3}
\end{align}
where $\mathbf{v}^{(\max)}_{i}$ in
(\ref{eq:data_rate_single_user_bound2}) is defined by
\begin{align}
\mathbf{v}^{(\max)}_{i} &= \arg \max\bigg\{ \left\|
\mathbf{v}^{[i',j']}\right\|^2: i'\in \mathcal{K}\setminus i, j'\in
\mathcal{S}\bigg\},
\end{align}
$\mathcal{S} \triangleq \{1, \ldots, S\}$, and $I^{[i,j]}$ in
(\ref{eq:data_rate_single_user_bound3}) is defined by
\begin{align}
I^{[i,j]} \triangleq \sum_{k=1, k\neq i}^{K} \sum_{s=1}^{S}\left\|
\mathbf{f}_{k}^{[i,j]}\right\|^2 \cdot \mathsf{SNR}.
\end{align}
Here, $\mathbf{v}_i^{(\max)}$ is fixed for given channel instance,
because $\mathbf{v}^{[i,j]}$ is determined by
$\mathbf{H}_i^{[i,j]}$, $j=1, \ldots, S$. Recalling that the indices
of the selected users are $(1, \ldots, S)$ for all cells, we can
expect the DoF of 1 for each user if and only if for some $0 \le
\epsilon< \infty$,
\begin{equation}
I^{[i,j]} < \epsilon, \hspace{10pt} \forall j \in \mathcal{S}, i\in
\mathcal{K}.
\end{equation}


To maximize the achievable DoF, we aim to minimize the
sum-interference $\sum_{i=1}^{K} \sum_{j=1}^{S}I^{[i,j]}$ through
receive beamforming at the users. Since $I^{[i,j]} = \sum_{s=1}^{S}
\eta^{[i,j]} \mathsf{SNR}$, we have
\begin{equation} \label{eq:sum_interference_equiv}
\sum_{i=1}^{K} \sum_{j=1}^{S}I^{[i,j]} = S\sum_{i=1}^{K}
\sum_{j=1}^{S} \eta^{[i,j]}\mathsf{SNR}.
\end{equation}
This implies that the collection of distributed effort to minimize
$\eta^{[i,j]}$ at the users can reduce the sum of received
interference. Therefore, each user finds the beamforming vector that
minimizes $\eta^{[i,j]}$ from
\begin{align}
\mathbf{u}^{[i,j]} &= \arg \min_{\mathbf{u}} \eta^{[i,j]} =  \arg \min_{\mathbf{u}} \sum_{k=1, k\neq i}^{K} \left\| \mathbf{u}^{\mathsf{H}}\mathbf{H}_{k}^{[i,j]} \mathbf{P}_k \right\|^2 \\
& = \arg \min_{\mathbf{u}} \left\| \mathbf{G}^{[i,j]}\mathbf{u}
\right\|^2,
\end{align}
where
\begin{align} \label{eq:G_def}
\mathbf{G}^{[i,j]} &\triangleq \Bigg[ \left(\mathbf{H}_{1}^{[i,j]}\mathbf{P}_1\right), \ldots, \left(\mathbf{H}_{i-1}^{[i,j]}\mathbf{P}_{i-1}\right),   \nonumber \\
&
\hspace{20pt}\left(\mathbf{H}_{i+1}^{[i,j]}\mathbf{P}_{i+1}\right),
\ldots,
\left(\mathbf{H}_{K}^{[i,j]}\mathbf{P}_{K}\right)\Bigg]^{\mathsf{H}}
\in \mathbb{C}^{(K-1)S \times L}.
\end{align}

Let us denote the singular value decomposition of
$\mathbf{G}^{[i,j]}$ as
\begin{equation} \label{eq:G_SVD}
\mathbf{G}^{[i,j]} =
\boldsymbol{\Omega}^{[i,j]}\boldsymbol{\Sigma}^{[i,j]}{\mathbf{V}^{[i,j]}}^{\mathsf{H}},
\displaybreak[0]
\end{equation}
where $\boldsymbol{\Omega}^{[i,j]}\in \mathbb{C}^{(K-1)S\times L}$
and $\mathbf{V}^{[i,j]}\in \mathbb{C}^{L\times L}$ consist of $L$
orthonormal columns, and $\boldsymbol{\Sigma}^{[i,j]} =
\textrm{diag}\left( \sigma^{[i,j]}_{1}, \ldots,
\sigma^{[i,j]}_{L}\right)$, where $\sigma^{[i,j]}_{1}\ge \cdots
\ge\sigma^{[i,j]}_{L}$. \pagebreak[0] Then, the optimal
$\mathbf{u}^{[i,j]}$ is determined as
\begin{equation} \label{eq:W_SVD}
\mathbf{u}^{[i,j]}= \mathbf{v}^{[i,j]}_{L},
\end{equation}
where $\mathbf{v}^{[i,j]}_{L}$ is the $L$-th column of
$\mathbf{V}^{[i,j]}$. With this choice the scheduling metric is
simplified to
\begin{equation} \label{eq:LIF_beamforming_simple}
\eta^{[i,j]} = {\sigma^{[i,j]}_{L}}^2.
\end{equation}
Since each column of $\mathbf{P}_k$ is isotropically and
independently distributed, each element of the effective
interference channel matrix $\mathbf{G}^{[i,j]}$ is i.i.d. complex
Gaussian with zero mean and unit variance.

%
%
%
%
We start with the following lemma for the probabilistic interference
level of the ODIA, which shall be frequently used in the sequel.
\begin{lemma} \label{lemma:CDF_scaling}
The sum-interference remains constant with high probability for
increasing SNR, that is,
\begin{align}
\label{eq:P_def}\lim_{\textsf{SNR} \rightarrow \infty}
 \mathcal{P}&\triangleq \lim_{\textsf{SNR}\rightarrow \infty} \textrm{Pr} \Bigg\{\sum_{i=1}^{K}\sum_{j=1}^{S} I^{[i,j]} \le \epsilon \Bigg\}=1 \displaybreak[0]
\end{align}
for any $0<\epsilon<\infty$, if
\begin{equation}
N = \omega\left( \mathsf{SNR}^{(K-1)S-L+1} \right).
\end{equation}
\end{lemma}
\begin{proof}
Since the cumulative density function (CDF) of $\eta^{[i,j]}$ is the
same as that of the scheduling metric of the MIMO IMAC
\cite{H_Yang13_TWC}, the lemma can be readily proved by following
the footsteps of Lemma 1 and Theorem 2 of \cite{H_Yang13_TWC}.
\end{proof}

Now, the following theorem establishes the DoF achievability of the
proposed ODIA.

\begin{theorem}[User scaling law] \label{theorem:DoF}
The proposed ODIA scheme with the scheduling metric
(\ref{eq:LIF_beamforming_simple}) achieves the optimal $KS$ DoF for
given $S$ with high probability  if
\begin{equation} \label{eq:N_scaling}
N=\omega\left(\textsf{SNR}^{(K-1)S-L+1}\right).
\end{equation}
\end{theorem}

\begin{proof}
If the sum-interference remains constant with probability
$\mathcal{P}$ as defined in (\ref{eq:P_def}), the achievable rate in
(\ref{eq:data_rate_single_user_bound3}) can be further bounded by
\begin{align}
R^{[i,j]} & \ge\left[ \log_2\left( \mathsf{SNR}\right) + \log_2
\left(\frac{1}{\mathsf{SNR}}+ \frac{ \frac{\gamma^{[i,j]}}{\left\|
\mathbf{v}^{(\max)}_{i}\right\|^2}}{ \frac{S}{\left\|
\mathbf{v}^{(\max)}_{i}\right\|^2}+ \epsilon } \right) \right] \cdot
\mathcal{P}, \label{eq:data_rate_single_user_bound5}
\end{align}
for any $0\le \epsilon < \infty$. Thus, the achievable DoF can be
bounded by
\begin{equation} \label{eq:DoF_SVD_LB}
\textrm{DoF} = \lim_{\textsf{SNR} \rightarrow \infty}
\frac{\sum_{i=1}^{K}\sum_{j=1}^{S}R^{[i,j]}}{\log \textsf{SNR}} \ge
\lim_{\textsf{SNR} \rightarrow \infty}
 KS \cdot \mathcal{P}.
\end{equation}
From Lemma \ref{lemma:CDF_scaling}, it is immediate to show that $\mathcal{P}$ tends to 1, and hence $KS$ DoF is achievable if $N = \omega\left( \mathsf{SNR}^{(K-1)S-L+1}\right)$, which proves the theorem. %
%
%
\end{proof}

Compared to the previous results of $N=\omega\left(
\mathsf{SNR}^{KS-L}\right)$ \cite{W_Shin12_IEICE,H_Nguyen13_arXiv,
J_Lee13_arXiv}, the exponent of SNR is reduced by $S-1$ using the
proposed ODIA, owing to perfect CSI of the selected users at each
BS, resulting in a slightly increased overhead. The essence of the
ODIA is that the design of the precoder $\mathbf{V}_i$ can be
decoupled from the design of the receive beamforming vector
$\mathbf{u}^{[i,j]}$, because the scheduling metric $\eta^{[i,j]}$
is calculated at the user side in a distributed fashion without the
knowledge of $\mathbf{V}_i$. Even with this decoupled approach,
interference can still be minimized due to the cascaded precoder
design. As a result, optimal DoF can be achieved without any
iterative precoder and receive beamforming vector optimization as
done in \cite{C_Suh11_TC}. In addition, the proposed ODIA applies to
arbitrary $M$, $L$, and $K$, whereas the optimal DoF is achievable
only in a few special cases in the scheme proposed in
\cite{C_Suh11_TC}.

The following remark discusses the uplink and downlink duality
within the OIA framework.
\begin{remark}[Uplink-downlink duality] \label{remark:up_down_duality}
The same scaling condition of $N=\omega\left(
\mathsf{SNR}^{K(S-1)-L+1}\right)$ was achieved to obtain $KS$ DoF in
the uplink interference channel \cite{H_Yang13_TWC}. Therefore,
Theorem \ref{theorem:DoF} implies that a duality holds true for the
uplink and downlink OIA frameworks in terms of the user scaling law.
\end{remark}

The user scaling law characterizes the trade-off between the
asymptotic DoF and number of users, i.e., the more number of users,
the more achievable DoF. In addition, we relate the derived user
scaling law to the interference decaying rate with respect to $N$
for given SNR. We start with the following lemma.
\begin{lemma}\label{lemma:CDF_decay}
The interference decaying rate of a selected $j$th user in the
$i$-th cell with respect to $N$ is given by
\begin{equation}
E\left\{\frac{1}{\eta^{[i,j]}} \right\} \ge \Theta\left(
N^{1/((K-1)S-L+1)}\right).
\end{equation}
Here, $f(x) = \Theta(g(x))$ if $f(x) = O(g(x))$ and $g(x) =
O(f(x))$.
\end{lemma}
\begin{proof}
The lemma can be shown by following the footsteps of the proof of
\cite[Theorem 3]{J_Jose12_Allerton}. The detailed proof is provided
in Appendix \ref{app:interf_decay}.
\end{proof}

\begin{theorem}[Interference decaying rate] \label{theorem:scaling_decay}
If the user scaling condition to achieve a target DoF is given by $N
= \omega \left( \textsf{SNR}^{\tau'}\right)$ for some $\tau'>0$,
then the interference decaying rate is given by
\begin{align}
 E\left\{\frac{1}{\eta^{[i,j]}} \right\} \ge \Theta\left( N^{1/\tau'}\right).
\end{align}
\end{theorem}
\begin{proof}
Since both the user scaling law and interference decaying rate are
determined by the tail CDF of the scheduling metric, the theorem can
be readily proved by the proofs of Theorem \ref{theorem:DoF} and
Lemma \ref{lemma:CDF_decay}.
\end{proof}

%

Therefore, from Theorem \ref{theorem:scaling_decay}, the user
scaling law also provides an insight on the interference decaying
rate with respect to $N$ for given SNR; that is, the smaller SNR
exponent of the user scaling law, the faster interference decreasing
rate with respect to $N$.

\vspace{20pt}
%
%
%
%

\section{ODIA with Limited feedback} \label{SEC:OIA_limited}
In the proposed ODIA scheme, the effective channel vectors
(${\mathbf{u}^{[i,j]}}^{\mathsf{H}}\mathbf{H}^{[i,j]}_{i}\mathbf{P}_i$)
in (\ref{eq:effective_CH}) can be fed back to the corresponding BS
using pilots rotated by the effective channels \cite{L_Choi04_TWC}.
However, this analog feedback requires two consecutive pilot phases
for each user: regular pilot for uplink channel estimation and
analog feedback for effective channel estimation. Hence, pilot
overhead grows with respect to the number of users in the network.
As a result, in practical systems with massive users, it is more
preferable to follow the widely-used limited feedback approach
\cite{D_Love03_TIT}, in which effective channels are fed back using
codebooks.

For limited feedback of effective channel vectors, we define the
codebook by
\begin{equation}
\mathcal{C}_f = \left\{ \mathbf{c}_{1}, \ldots,
\mathbf{c}_{N_f}\right\},
\end{equation}
where $N_f$ is the codebook size and $\mathbf{c}_k\in \mathbb{C}^{S
\times 1}$ is a unit-norm codeword, i.e., $\|\mathbf{c}_i\|^2=1$.
Hence, the number of feedback bits used is given by
\begin{equation}
n_f = \lceil\log_2 N_f \rceil  (\textrm{bits})
\end{equation}
For the effective channel ${\mathbf{f}_{i}^{[i,j]}}^{\mathsf{H}} =
{\mathbf{u}^{[i,j]}}^{\mathsf{H}} \mathbf{H}^{[i,j]}_{i}
\mathbf{P}_i,$ each user quantizes the normalized effective channel
for given $\mathcal{C}_f$ from
\begin{align}
\mathbf{f}_{i}^{[i,j]} = \arg \max_{ \{\mathbf{w} = \mathbf{c}_k:
1\le k \le N_f\}} \frac{\left|
{\mathbf{f}_{i}^{[i,j]}}^{\mathsf{H}}\mathbf{w}\right|^2}{\left\|
\mathbf{f}_{i}^{[i,j]}\right\|^2}.
\end{align}
Now, the user feeds back three types of information: 1) index of
$\mathbf{f}_{i}^{[i,j]}$, 2) channel gain of $\left\|
\mathbf{f}_{i}^{[i,j]}\right\|^2$, and 3) scheduling metric
$\eta^{[i,j]}$. Note that the channel gains and scheduling metrics
are real scalar values, and thus can be accurately fed back as
uplink data. Then, BS $i$ constructs the quantized effective channel
vectors $\hat{\mathbf{f}}^{[i,j]}$ from
\begin{align}
\hat{\mathbf{f}}^{[i,j]} \triangleq \left\|
\mathbf{f}_{i}^{[i,j]}\right\|^2 \cdot \mathbf{f}_{i}^{[i,j]},
\hspace{10pt}i=1, \ldots, S,
\end{align}
and the precoding matrix $\hat{\mathbf{V}}_i$ from
\begin{align} \label{eq:V_hat}
\hat{\mathbf{V}}_i = \hat{\mathbf{F}}_i^{-1} \boldsymbol{\Gamma}_i,
\end{align}
where $\boldsymbol{\Gamma}_i = \textrm{diag} \left(
\sqrt{\gamma^{[i,1]}}, \ldots, \sqrt{\gamma^{[i,S]}}\right)$ and
$\hat{\mathbf{F}}_i = \left[ \hat{\mathbf{f}}^{[i,1]}, \ldots,
\hat{\mathbf{f}}^{[i,S]}\right]^{\mathsf{H}}$.
%

With limited feedback, the received signal vector after receive
beamforming is written by
\begin{align}\label{eq:rec_vector_after_BF_limited}
\tilde{y}^{[i,j]} &= {\mathbf{f}_{i}^{[i,j]}}^{\mathsf{H}}\hat{\mathbf{V}}_i \mathbf{x}_i + \cdot \sum_{k=1, k\neq i}^{K} {\mathbf{f}_{k}^{[i,j]}}^{\mathsf{H}} \hat{\mathbf{V}}_k \mathbf{x}_k \nonumber \\
& \hspace{50pt}+ {\mathbf{u}^{[i,j]}}^{\mathsf{H}} \mathbf{z}^{[i,j]}  \\
& = \sqrt{\gamma^{[i,j]}}x^{[i,j]} + \underbrace{\left( {\mathbf{f}_{i}^{[i,j]}}^{\mathsf{H}}\hat{\mathbf{V}}_i \mathbf{x}_i- \sqrt{\gamma^{[i,j]}}x^{[i,j]}\right)}_{\textrm{residual intra-cell interference}} \nonumber \\
& \hspace{20pt}+ \sum_{k=1, k\neq i}^{K}
{\mathbf{f}_{k}^{[i,j]}}^{\mathsf{H}} \hat{\mathbf{V}}_k
\mathbf{x}_k + {\mathbf{u}^{[i,j]}}^{\mathsf{H}} \mathbf{z}^{[i,j]},
\end{align}
where the residual intra-cell interference is non-zero due to the
quantization error in $\hat{\mathbf{V}}_i$.

It is important to note that the residual intra-cell interference is
a function of $\hat{\mathbf{V}}_i$, which includes other users'
channel information, and thus each user treats this term as
unpredictable noise and calculates only the inter-cell interference
for the scheduling metric as in (\ref{eq:eta}); that is, the
scheduling metric is not changed for the ODIA with limited feedback.

The following theorem establishes the user scaling law for the ODIA
with limited feedback.
\begin{theorem} \label{th:codebook}
The ODIA with a Grassmannian or random codebook achieves the same
user scaling law of the ODIA with perfect CSI described in Theorem
\ref{theorem:DoF}, if
\begin{equation} \label{eq:nf_cond0}
n_f =\omega\left( \log_2 \mathsf{SNR} \right).
\end{equation}
That is, $KS$ DoF is achievable with high probability if $N=\omega
\left( \mathsf{SNR}^{(K-1)S-L+1}\right)$ and (\ref{eq:nf_cond0})
holds true.
\end{theorem}
\begin{proof}
Without loss of generality, the quantized effective channel vector
$\hat{\mathbf{f}}^{[i,j]}$ can be decomposed as
\begin{align} \label{eq:f_hat2}
\hat{\mathbf{f}}^{[i,j]} &= \left\| \mathbf{f}_{i}^{[i,j]}\right\|^2 \cdot \mathbf{f}_{i}^{[i,j]} \nonumber \\
& =  \sqrt{1-{d^{[i,j]}}^2}\cdot\mathbf{f}_{i}^{[i,j]}+
d^{[i,j]}\left\|\mathbf{f}_{i}^{[i,j]}\right\|^2
\left(\mathbf{t}^{[i,j]}\right),
\end{align}
where $\mathbf{t}^{[i,j]}$ is a unit-norm vector i.i.d. over
$\textrm{null}\left( \mathbf{f}_{i}^{[i,j]}\right)$
\cite{N_Jindal06_TIT}. At this point, we consider the worse
performance case where each user finds $\hat{\mathbf{f}}^{[i,j]}$
such that with a slight abuse of notation
\begin{align} \label{eq:f_hat3}
\hat{\mathbf{f}}^{[i,j]} =
\sqrt{1-{d^{\max}_i}^2}\cdot\mathbf{f}_{i}^{[i,j]}+ d^{\max}_i\nu_i
\cdot\mathbf{t}^{[i,j]},
\end{align}
where
\begin{equation}
d^{\max}_i = \max \left\{d^{[i,1]}, \ldots, d^{[i,S]} \right\},
\end{equation}
\begin{equation}
\nu_i = \max \left\{ \left\|\mathbf{f}_{i}^{[i,j]}\right\|^2, j=1,
\ldots, S \right\}.
\end{equation}
Note that more quantization error only degrades the achievable rate,
and hence the quantization via (\ref{eq:f_hat3}) yields a
performance lower-bound. Inserting (\ref{eq:f_hat3}) to
(\ref{eq:V_hat}) gives us
\begin{align} \label{eq:V_hat2}
\hat{\mathbf{V}}_i = \left( \sqrt{1- {d^{\max}_i}^2}\mathbf{F}_i +
d^{\max}_i\nu_i\mathbf{T}_i \right)^{-1} \boldsymbol{\Gamma}_i,
\end{align}
where ${\mathbf{F}}_i = \left[ \mathbf{f}_{i}^{[i,1]}, \ldots,
\mathbf{f}_{i}^{[i,S]}\right]^{\mathsf{H}}$ and ${\mathbf{T}}_i =
\left[ {\mathbf{t}}^{[i,1]}, \ldots,
{\mathbf{t}}^{[i,S]}\right]^{\mathsf{H}}$.

The Taylor expansion of $\left( \sqrt{1-{d^{\max}_i}^2}\mathbf{F}_i
+ d^{\max}_i\nu_i\mathbf{T} \right)^{-1}$ in (\ref{eq:V_hat}) gives
us
\begin{align} \label{eq:Taylor}
&\left( \sqrt{1-{d^{\max}_i}^2}\mathbf{F}_i + d^{\max}_i\nu_i\mathbf{T}_i \right)^{-1} \nonumber \\
& \hspace{20pt}= \mathbf{F}_i^{-1} - \mathbf{F}_i^{-1}\mathbf{T}_i
\mathbf{F}_i^{-1} \nu_id^{\max}_i +\sum_{k=2}^{\infty} \mathbf{A}_k
\left(d^{\max}_i\right)^k,
\end{align}
where $\mathbf{A}_k$ is a function of $\mathbf{F}_i$ and
$\mathbf{T}_i$. Thus, $\hat{\mathbf{V}}_i$ can be written by
\begin{align} \label{eq:V_hat_final}
\hat{\mathbf{V}}_i = \mathbf{F}_i^{-1}\boldsymbol{\Gamma}_i -
d^{\max}_i\nu_i\mathbf{F}_i^{-1}\mathbf{T}_i
\mathbf{F}_i^{-1}\boldsymbol{\Gamma}_i +\sum_{k=2}^{\infty}
\left(d^{\max}_i\right)^k\mathbf{A}_k\boldsymbol{\Gamma}_i
\end{align}

Inserting (\ref{eq:V_hat_final}) to
(\ref{eq:rec_vector_after_BF_limited}) yields
\begin{align}\label{eq:rec_vector_after_BF_limited3}
\tilde{y}^{[i,j]} &= \sqrt{\gamma^{[i,j]}}x^{[i,j]} \nonumber \\
& \underbrace{-d^{\max}_i\nu_i{\mathbf{t}^{[i,j]}}^{\textsf{H}}\mathbf{F}_i^{-1}\boldsymbol{\Gamma}_i \mathbf{x}_i + \sum_{k=2}^{\infty}  \left(d^{\max}_i\right)^k{\mathbf{f}_{i}^{[i,j]}}^{\mathsf{H}}\mathbf{A}_k \boldsymbol{\Gamma}_i \mathbf{x}_i }_{\textrm{residual intra-cell interference}} \nonumber \\
& \hspace{0pt}+ \sum_{k=1, k\neq i}^{K}
{\mathbf{f}_{k}^{[i,j]}}^{\mathsf{H}} \hat{\mathbf{V}}_k
\mathbf{x}_k + {\mathbf{u}^{[i,j]}}^{\mathsf{H}} \mathbf{z}^{[i,j]}.
\end{align}
Consequently,  the rate $R^{[i,j]}$ in
(\ref{eq:data_rate_single_user}) is given by
\begin{align} \label{eq:data_rate_single_user3}
R^{[i,j]}= \log_2 \left( 1+ \frac{ \gamma^{[i,j]} }{ \frac{S+
\Delta^{[i,j]}}{\mathsf{SNR}}+ \sum_{k\neq i}^{K} \sum_{s=1}^{S}
\left| {\mathbf{f}_{k}^{[i,j]}}^{\mathsf{H}}
\mathbf{v}^{[k,s]}\right|^2 } \right),
\end{align}
where
\begin{align} \label{eq:Delta}
\Delta^{[i,j]} = \left(d^{\max}_i\right)^2\delta_1\cdot \mathsf{SNR}
+ \sum_{k=2}^{\infty}  \left(d^{\max}_i\right)^{2k}\delta_k\cdot
\mathsf{SNR},
\end{align}
\begin{equation}
\delta_1 =
\left(\nu_i^2{\mathbf{t}^{[i,j]}}^{\textsf{H}}\mathbf{F}_i^{-1}{\boldsymbol{\Gamma}_i}^2\mathbf{F}_i^{-\mathsf{H}}\mathbf{t}^{[i,j]}\right),
\end{equation}
\begin{equation}
\delta_k = \left({\mathbf{f}_{i}^{[i,j]}}^{\mathsf{H}}\mathbf{A}_k
\boldsymbol{\Gamma}_i^2\mathbf{A}_k^{\mathsf{H}}
\mathbf{f}_{i}^{[i,j]}\right)
\end{equation}
As in (\ref{eq:data_rate_single_user_bound}) to
(\ref{eq:data_rate_single_user_bound3}), the achievable rate can be
bounded by
\begin{align} \label{eq:R_bound_limitedFB}
R^{[i,j]} & \ge\left[ \log_2\left( \mathsf{SNR}\right) + \log_2
\left( \frac{1}{\mathsf{SNR}}+ \frac{ \frac{\gamma^{[i,j]}}{\left\|
\mathbf{v}^{(\max)}_{i}\right\|^2}}{ \frac{1}{\left\|
\mathbf{v}^{(\max)}_{i}\right\|^2}+ 2\epsilon } \right) \right]
\cdot \mathcal{P}',
\end{align}
where
\begin{align}
\label{eq:P_def2}
\mathcal{P}'&\triangleq  \textrm{Pr} \Bigg\{\left(\sum_{i=1}^{K}\sum_{j=1}^{S} I^{[i,j]} \le \epsilon\right) \& \left(\Delta^{[i,j]}/\left\| \mathbf{v}^{(\max)}_{i}\right\|^2 \le \epsilon\right), \nonumber \\
& \hspace{80pt} \forall i\in \mathcal{K}, j\in \mathcal{S}  \Bigg\} \displaybreak[0] \\
\label{eq:P_def3}& = \textrm{Pr} \Bigg\{\sum_{i=1}^{K}\sum_{j=1}^{S} I^{[i,j]} \le \epsilon, \forall i\in \mathcal{K}, j\in \mathcal{S}  \Bigg\}\nonumber \\
& \hspace{50pt}\times \textrm{Pr} \Bigg\{\Delta^{[i,j]} \le
\epsilon', \forall i\in \mathcal{K}, j\in \mathcal{S}  \Bigg\},
\end{align}
where $\epsilon' \triangleq \epsilon\cdot\left\|
\mathbf{v}^{(\max)}_{i}\right\|^2$.
 Here, (\ref{eq:P_def3}) follows from the fact that the inter-cell interference $I^{[i,j]}$ and residual intra-cell interference  $\Delta^{[i,j]}$ are independent each other. Note also that the level of residual intra-cell interference does not affect the user selection and is determined only by the codebook size $N_f$. Hence, the user selection result does not change for different $N_f$.

The achievable DoF is given by
 \begin{align}
 \textrm{DoF} \ge \lim_{\textsf{SNR} \rightarrow \infty}
 KS \cdot \mathcal{P}'.
 \end{align}
 If $N=\omega\left(\textsf{SNR}^{(K-1)S-L+1}\right)$, the first term of (\ref{eq:P_def3}) tends to 1 according to Lemma \ref{lemma:CDF_scaling}.
Thus, the maximum DoF can be obtained if and only if $\Delta^{[i,j]}
\le \epsilon'$ for all selected users for increasing SNR.

In Appendix \ref{app:th_codebook}, it is shown that $\Delta^{[i,j]}
\le \epsilon'$ for all selected users if $n_f =\omega\left( \log_2
\mathsf{SNR}\right)$ for both Grassmannian and random codebooks.
Therefore, if $N=\omega\left(\textsf{SNR}^{(K-1)S-L+1}\right)$ and
$n_f =\omega\left( \log_2 \mathsf{SNR}\right)$, $\mathcal{P}'$ in
(\ref{eq:P_def3}) tends to 1, which proves the theorem.
\end{proof}

From Theorem \ref{th:codebook}, the minimum number of feedback bits
$n_f$ is characterized to achieve the optimal $KS$ DoF, which
increases with respect to $\log_2(\mathsf{SNR})$. It is worthwhile
to note that the results are the same for the Grassmannian and
random codebooks.  In the previous works on limited feedback
systems, the performance analysis was focused on the average SNR or
the average rate loss \cite{C_Au-Yeung09_TWC}. In an average sense,
the Grassmannian codebook is in general outperforms the random
codebook. However, our scheme focuses on the asymptotic codebook
performance for given channel instance for increasing SNR, and it
turned out that this asymptotic behaviour is the same for the two
codebooks. In fact, this result agrees with the previous works e.g.,
\cite{B_Khoshnevis11_Thesis}, in which the performance gap between
the two codebooks was shown to be negligible as $n_f$ increases
through computer simulations.

We conclude this section by providing the following comparison to
the well-known conventional result on limited feedback systems.
\begin{remark}
For the MIMO broadcast channel with limited feedback, where the
transmitter has $L$ antennas and employs the random codebook, it was
shown \cite{N_Jindal06_TIT} that the achievable rate loss for each
user, denoted by $\Delta R$, due to the finite size of the codebook
is lower-bounded by
\begin{equation}
\Delta < \log_2 \left(1+\textrm{SNR} \cdot 2^{-n_f/(L-1)} \right).
\end{equation}
Thus, to achieve the maximum 1 DoF for each user, or to make the
rate loss negligible as the SNR increases, the term $\textrm{SNR}
\cdot 2^{-n_f/(L-1)}$ should remain constant for increasing SNR.
That is, $n_f$ should scale faster than $(L-1)\log_2
(\textrm{SNR})$. Though the system is different, our results of
Theorem \ref{th:codebook} are consistent with this previous result.
\end{remark}

\section{Spectrally Efficient ODIA (SE-ODIA)} \label{SEC:Threhold_ODIA}
In this section, we propose a spectrally efficient OIA (SE-ODIA)
scheme and show that the proposed SE-ODIA achieves the optimal
multiuser diversity gain $\log \log N$. For the DoF achievability,
it was enough to design the user scheduling in the sense to minimize
inter-cell interference. However, to achieve optimal multiuser
diversity gain, the gain of desired channels also needs to be
considered in user scheduling. The overall procedure of the SE-ODIA
follows that of the ODIA described in Section \ref{SEC:OIA} except
the the third stage `User Scheduling'. In addition, we assume the
perfect feedback of the effective desired channels
${\mathbf{u}^{[i,j]}}^{\mathsf{H}}\mathbf{H}_i^{[i,j]}\mathbf{P}_i$
for the SE-ODIA. We incorporate the semiorthogonal user selection
algorithm proposed in \cite{T_Yoo06_JSAC} to the ODIA framework
taking into consideration inter-cell interference. Specifically, the
algorithm for the user scheduling at the BS side is as follows:
\begin{itemize}
\item Step 1: Initialization:
\begin{align}
\mathcal{N}_1& = \{1, \ldots, N\}, \hspace{10pt} s=1
\end{align}
\item Step 2: For each user $j\in \mathcal{N}_s$ in the $i$-th cell, the $s$-th orthogonal projection vector, denoted by $\tilde{\mathbf{b}}_{s}^{[i,j]}$, for given $\left\{ \mathbf{b}_{1}^{[i]}, \ldots, \mathbf{b}_{s-1}^{[i]} \right\}$ is calculated from:
\begin{align}
\tilde{\mathbf{b}}^{[i,j]}_s &= \mathbf{f}_{i}^{[i,j]} -
\sum_{s'=0}^{s-1} \frac{{\mathbf{b}_{s'}^{[i]}}^{\mathsf{H}}
\mathbf{f}_{i}^{[i,j]}}{\|\mathbf{b}_{s'}^{[i]}\|^2}\mathbf{b}_{s'}^{[i]}
\end{align}
Note that if $s=1$, $\tilde{\mathbf{b}}_{1}^{[i,j]} =
\mathbf{f}_{i}^{[i,j]}$.
\item Step 3: For the $s$-th user selection, a user is selected at random from the user pool $\mathcal{N}_s$ that satisfies the following two conditions:
\begin{align}
\label{eq:C}\mathsf{C}_1:& \eta^{[i,j]} \le \eta_I,
\hspace{10pt}\mathsf{C}_2: \|\tilde{\mathbf{b}}_{s}^{[i,j]}\|^2 \ge
\eta_D
\end{align}
Denote the index of the selected user by $\pi(s)$ and define
\begin{equation}
\mathbf{b}_{s}^{[i]} = \tilde{\mathbf{b}}^{[i,\pi(s)]}_s.
\end{equation}
\item Step 4: If $s < S$, then find the $(s+1)$-th user pool $\mathcal{N}_{s+1}$ from:
\begin{align}
\mathcal{N}_{s+1}& = \left\{j:j \in \mathcal{N}_{s}, j \neq \pi(s), \frac{\left|{\mathbf{f}_{i}^{[i,j]}}^{\mathsf{H}} \mathbf{b}_{s}^{[i]}\right|}{\| \mathbf{f}_{i}^{[i,j]}\| \|\mathbf{b}_{s}^{[i]}\|}  <\alpha\right\},\hspace{10pt} \\
 s &= s+1,
\end{align}
where $\alpha>0$ is a positive constant. Repeat Step 2 to Step 4
until $s=S$.
\end{itemize}

To show the SE-ODIA achieves the optimal multiuser diversity gain,
we start with the following lemma for the bound on
$|\mathcal{N}_s|$.
\begin{lemma}\label{lemma:N_card}
The cardinality of $\mathcal{N}_s$ can be bounded by
\begin{align}
|\mathcal{N}_s| & \gtrsim N \cdot \alpha^{2(S-1)}.
\end{align}
The approximated inequality becomes tight as $N$ increases.
\end{lemma}
\begin{proof}
See Appendix \ref{app:Ns_cardinality}.
\end{proof}

We also introduce the following useful lemma.
\begin{lemma}\label{lemma:quadratic}
If $x \in \mathbb{C}^{M \times 1}$ has its element i.i.d. according
to $\mathcal{CN}(0, \sigma^2)$ and $\mathbf{A}$ is an idempotent
matrix of rank $r$ (i.e., $\mathbf{A}^2= \mathbf{A}$), then
$\mathbf{x}^{\mathsf{H}} \mathbf{A} \mathbf{x}/\sigma^2$ has a
Chi-squared distribution with $2r$ degrees-of-freedom.
\end{lemma}
\begin{proof}
See \cite{G_Seber03_Book}.
\end{proof}

In addition, the following lemma on the achievable rate of the
SE-ODIA will be used to show the achievability of optimal multiuser
diversity gain.
\begin{lemma}\label{lemma:effective_gain}
For the $j$-th selected user in the $i$-th cell, the achievable rate
is bounded by
\begin{align} \label{eq:data_rate_single_user4}
R^{[i,j]}\ge \log_2 \left( 1+ \frac{ \frac{\left\|
\mathbf{b}_{j}^{[i]} \right\|^2}{1+ \frac{(S-1)^4
\alpha^2}{1-(S-1)\alpha^2}} }{ \frac{S}{\mathsf{SNR}} + \sum_{k\neq
i}^{K} \sum_{s=1}^{S} \left|
{\mathbf{f}_{k}^{[i,j]}}^{\mathsf{H}}\mathbf{v}^{[k,s]}\right|^2 }
\right).
\end{align}
\end{lemma}
\begin{proof}
Since the chosen channel vectors are not perfectly orthogonal, there
is degradation in the effective channel gain $\gamma^{[i,j]}$.
Specifically, for the $j$-th selected user in the $i$-th cell, we
have
\begin{align}\label{eq:gamma_j_bound}
\gamma^{[i,j]} &= \frac{1}{\left[
\left(\mathbf{F}_i\mathbf{F}_i^{\mathsf{H}}\right)^{-1}
\right]_{j,j}} > \frac{\left\| \mathbf{b}_{j}^{[i]} \right\|^2}{1+
\frac{(S-1)^4 \alpha^2}{1-(S-1)\alpha^2}},
\end{align}
which follows from \cite[Lemma 2]{T_Yoo06_JSAC}. Inserting
(\ref{eq:gamma_j_bound}) to the sum-rate lower bound in
(\ref{eq:data_rate_single_user}) proves the lemma.
\end{proof}

Now the following theorem establishes the achievability of the
optimal multiuser diversity gain.

\begin{theorem}\label{theorem:MUD}
The proposed SE-ODIA scheme with
\begin{equation} \label{eq:eta_D_choice}
\eta_D = \epsilon_D \log \mathsf{SNR}
\end{equation}
\begin{equation}\label{eq:eta_I_choice}
\eta_I = \epsilon_I\mathsf{SNR}^{-1}
\end{equation}
for any $\epsilon_D, \epsilon_I >0$ achieves the optimal multiuser
diversity gain given by
\begin{equation}
R^{[i,j]} = \Theta\left( \log \left( \mathsf{SNR}\cdot \log
N\right)\right),
\end{equation}
with high probability for all selected users in the high SNR regime
if
\begin{align} \label{eq:N_scaling_MUD}
N = \omega \left(
\mathsf{SNR}^{\frac{(K-1)S-L+1}{1-(\epsilon_D/2)}}\right).
\end{align}
\end{theorem}

\begin{proof}
Amongst $|\mathcal{N}_s|$ users, there should exist at least one
user satisfying the conditions $\mathsf{C}_1$ and $\mathsf{C}_2$ to
make the proposed user scheduling for the SE-ODIA valid. Thus, we
first show the probability that there exist at least one valid user,
denoted by $\mathsf{p}_s$, converges to 1, for the $s$-th user
selection, if $N$ scales according to (\ref{eq:N_scaling_MUD}) with
the choices (\ref{eq:eta_D_choice}) and (\ref{eq:eta_I_choice}).

The probability that each user satisfies the two conditions is given
by $\textrm{Pr} \left\{ \mathsf{C}_1  \right\} \cdot \textrm{Pr}
\left\{ \mathsf{C}_2  \right\}$, because the two conditions are
independent of each other. Consequently, $\mathsf{p}_{s}$ is given
by
\begin{align}
\mathsf{p}_s &= 1- \left( 1- \textrm{Pr} \left\{ \mathsf{C}_1  \right\} \cdot \textrm{Pr} \left\{ \mathsf{C}_2  \right\} \right)^{|\mathcal{N}_s|} \\
\label{eq:p_s2}& \gtrsim 1- \left( 1- \textrm{Pr} \left\{
\mathsf{C}_1  \right\} \cdot \textrm{Pr} \left\{ \mathsf{C}_2
\right\} \right)^{N \cdot \alpha^{2(S-1)}}.
\end{align}
Note that each element of ${\mathbf{f}_{i}^{[i,j]}}^{\mathsf{H}} =
{\mathbf{u}^{[i,j]}}^{\mathsf{H}} \mathbf{H}^{[i,j]}_{i}
\mathbf{P}_i$ is i.i.d. according to $\mathcal{CN}(0,1)$, because
each column of $\mathbf{P}_i$ is a random orthogonal unit vector and
because ${\mathbf{u}^{[i,j]}}^{\mathsf{H}}$ is designed
independently of $\mathbf{H}_i^{[i,j]}$ and isotropically
distributed over a unit sphere. Thus,
${\mathbf{f}_{i}^{[i,j]}}^{\mathsf{H}} =
{\mathbf{u}^{[i,j]}}^{\mathsf{H}} \mathbf{H}^{[i,j]}_{i}
\mathbf{P}_i$ has its element i.i.d. according to
$\mathcal{CN}(0,1)$.

Let us define $\mathbf{P}$ by
\begin{align}
\mathbf{P} &\triangleq \left( \mathbf{I} - \sum_{s'=0}^{s-1} \frac{
\mathbf{b}_{s'}^{[i]}{\mathbf{b}_{s'}^{[i]}}^{\mathsf{H}}}{
\|\mathbf{b}_{s'}^{[i]}\|^2}\right),
\end{align}
which is a symmetric idempotent matrix with rank $(S-s+1)$. Since
$\mathbf{b}_{s}^{[i]} = \mathbf{P}\mathbf{f}_{i}^{[i,j]}$, from
Lemma \ref{lemma:quadratic}, $\left\|\mathbf{b}_{s}^{[i]}\right\|^2$
is a Chi-squared random variable with $2(S-s+1)$ degrees-of-freedom.

In Appendix \ref{app:MUD}, for $\eta_D>2$, we show that
\begin{equation} \label{eq:ps_conv}
\lim_{\mathsf{SNR}\rightarrow \infty} \mathsf{p}_s = 1,
\hspace{10pt} \textrm{if } N = \omega \left(
\mathsf{SNR}^{\frac{(K-1)S-L+1}{1-(\epsilon_D/2)}}\right).
\end{equation}

Now, given that there always exist at least one user that satisfies
the conditions $\mathsf{C}_1$ and $\mathsf{C}_2$, the achievable
sum-rate can be bounded from Lemma \ref{lemma:effective_gain} by
\begin{align} \label{eq:data_rate_single_user5}
R^{[i,j]}&\ge \log_2 \left( 1+ \frac{ \frac{\left\| \mathbf{b}_{j}^{[i]} \right\|^2}{1+ \frac{(S-1)^4 \alpha^2}{1-(S-1)\alpha^2}} \cdot \frac{1}{\left\|\mathbf{v}_i^{\max} \right\|^2}}{ \frac{S}{\mathsf{SNR}\left\|\mathbf{v}_i^{\max} \right\|^2}+ \sum_{k\neq i}^{K} \sum_{s=1}^{S} \left\| \mathbf{f}_{k}^{[i,j]}\right\|^2 } \right)\\
\label{eq:data_rate_single_user6}& \ge \log_2 \left( 1+ \frac{ \frac{\left\| \mathbf{b}_{j}^{[i]} \right\|^2}{1+ \frac{(S-1)^4 \alpha^2}{1-(S-1)\alpha^2}} \cdot \mathsf{SNR}/\left\|\mathbf{v}_i^{\max} \right\|^2}{ S/\left\|\mathbf{v}_i^{\max} \right\|^2+ KS\epsilon_I} \right)\\
\label{eq:data_rate_single_user7}& = \log_2 \left( 1+ \left\| \mathbf{b}_{j}^{[i]} \right\|^2\mathsf{SNR} \cdot \xi \right)\\
\label{eq:data_rate_single_user8}& \ge \log_2 \left( 1+
\epsilon_D(\log N)\cdot \mathsf{SNR} \right),
\end{align}
where (\ref{eq:data_rate_single_user6}) follows from  the fact that
the sum-interference for all selected users, given by
$\sum_{j=1}^{S}\sum_{i=1}^{K} \eta^{[i,j]}\mathsf{SNR}$ (See
(\ref{eq:sum_interference_equiv})), does not exceed $KS \epsilon_I $
by choosing $\eta_I = \epsilon_I \mathsf{SNR}^{-1}$. Furthermore,
$\xi$ is a constant given by
\begin{equation}
\xi = \frac{1}{\left\| \mathbf{v}_i^{\max}\right\|^2 \left( 1+
\frac{(S-1)^4
\alpha^2}{1-(S-1)\alpha^2}\right)\left(S/\left\|\mathbf{v}_i^{\max}
\right\|^2+ KS\epsilon_I\right)},
\end{equation}
and (\ref{eq:data_rate_single_user8}) follows from $\|
\mathbf{b}_{j}^{[i]} \|^2 \ge \eta_D  = \epsilon_D \log N$.
Therefore, the proposed SE-ODIA achieves the optimal multiuser
diversity gain $\log\log N$ in the high SNR regime, if $N = \omega
\left( \mathsf{SNR}^{\frac{(K-1)S-L+1}{1-(\epsilon_D/2)}}\right)$.
\end{proof}

Therefore, the optimal multiuser gain of $\log\log N$ is achieved
using the proposed SE-ODIA with the choices of
(\ref{eq:eta_D_choice}) and (\ref{eq:eta_I_choice}). Note that since
small $\epsilon_D$ suffices to obtain the optimal multiuser gain,
the condition on $N$ does not dramatically change compared with that
required to achieve $KS$ DoF (See Theorem \ref{theorem:DoF}).
Combining the results in Theorem \ref{theorem:DoF} and
\ref{theorem:MUD}, we can conclude the achievability of the optimal
DoF and multiuser gain as follows.
\begin{remark}
In fact, the ODIA described in Section \ref{SEC:OIA} can be
implemented using the SE-ODIA approach by choosing $\eta_D = 0$,
$\alpha = 1$, and $\eta_I^{[i]} = \min\left\{ \eta^{[i,1]}, \ldots,
\eta^{[i,N]}\right\}$, where $\eta_I^{[i]}$ denotes $\eta_I$ at the
$i$-th cell. In summary, the optimal $KM$ DoF and optimal multiuser
gain of $\log \log N$ can be achieved using the proposed ODIA
framework, if the number of users per cell increases according to $N
= \omega\left(
\mathsf{SNR}^{\frac{(K-1)M-L+1}{1-(\epsilon_D/2)}}\right)$ for any
$\epsilon_D>0$.
\end{remark}

\section{Numerical Results} \label{SEC:Sim}
In this section, we compare the performance of the proposed ODIA
with two conventional schemes which also utilize the multi-cell
random beamforming technique at BSs. First, we consider ``max-SNR"
technique, in which each user designs the receive beamforming vector
in the sense to maximize the desired signal power, and feeds back
the maximized signal power to the corresponding BS. Each BS selects
$S$ users who have higher received signal power. Second, ``min-INR"
technique is considered, in which each user performs receive
beamforming in order to minimize the sum of inter-cell interference
and intra-cell interference\cite{H_Nguyen13_arXiv, J_Lee13_arXiv}.
Hence, intra-cell interference does not vanish at users, while the
proposed ODIA perfectly eliminates it via transmit beamforming.
Specifically, from (\ref{eq:rec_vector_after_BF}), the $j$-th user
in the $i$-th cell should calculate the following $S$ scheduling
metrics
 \begin{align}
\eta^{[i,j]}_{\textrm{min-INR}, m} &= \underbrace{\left\|{\mathbf{u}^{[i,j],m}}^{\mathsf{H}} \mathbf{H}_i^{[i,j]}\tilde{\mathbf{P}}_{i,m} \right\|^2}_{\textrm{intra-cell interference}} \nonumber \\
& \hspace{20pt}+ \underbrace{\sum_{k=1, k\neq
i}^{K}\left\|{\mathbf{u}^{[i,j],m}}^{\mathsf{H}}
\mathbf{H}_k^{[i,j]}\mathbf{P}_{k} \right\|^2}_{\textrm{inter-cell
interference}},
\end{align}
for $m=1, \ldots, S$, where $\tilde{\mathbf{P}}_{i,m} \triangleq
\left[ \mathbf{p}_{1, i}, \ldots, \mathbf{p}_{m-1,i},
\mathbf{p}_{m+1, i}, \ldots, \mathbf{p}_{S,i} \right]$. For each
$m$, the receive beamforming vector $\mathbf{u}^{[i,j],m}$ is
assumed to be designed such that $\eta^{[i,j]}_{\textrm{min-INR},
m}$ is minimized. Each user feedbacks $S$ scheduling metrics to the
corresponding BS, and the BS selects the user having the minimum
scheduling metric for the $m$-th spatial stream, $m=1, \ldots, S$.
For more details about the min-INR scheme, refer to
\cite{H_Nguyen13_arXiv, J_Lee13_arXiv}.

Fig. \ref{fig:Interf_N} shows the sum-interference at all users for
varying number of users per cell, $N$, when $K=3$, $M=4$, $L=2$, and
SNR=$20$dB. The solid lines are obtained from Theorem
\ref{theorem:scaling_decay} with proper biases, and thus only the
slopes of the solid lines are relevant. The decaying rates of
sum-interference of the proposed ODIA are higher than those of the
min-INR scheme since intra-cell interference is perfectly eliminated
in the proposed ODIA. In addition, the interference decaying rates
of the proposed ODIA are consistent with the theoretical results of
Theorem \ref{theorem:scaling_decay}, which proves that the user
scaling condition derived in Theorem \ref{theorem:DoF} and the
interference bound in Theorem \ref{theorem:scaling_decay} are in
fact accurate and tight.

\begin{figure}[t]
\begin{center}
  \includegraphics[width=0.7\textwidth]{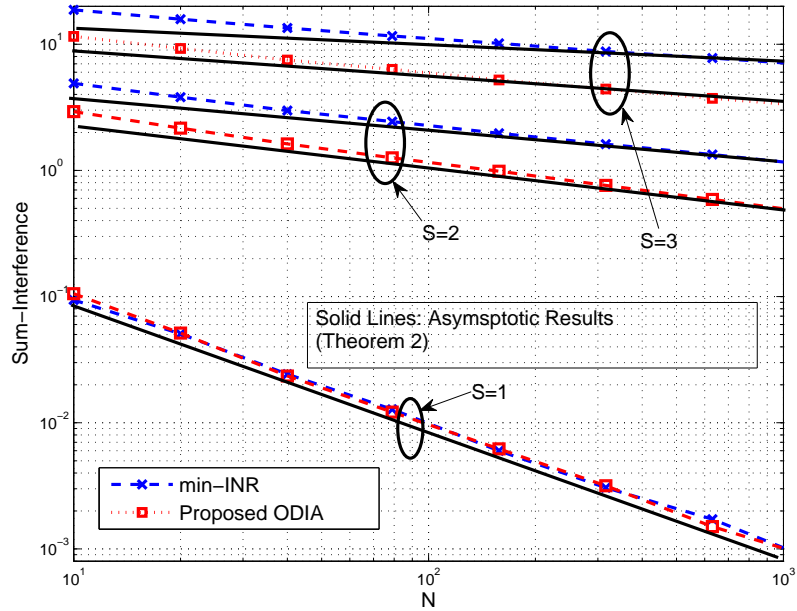}\\
  \caption{Normalized sum-interference vs. $N$ when $K=3$, $M=4$, $L=2$.}\label{fig:Interf_N}
  \end{center}
\end{figure}

To evaluate the sum-rates of the proposed ODIA schemes, the
parameters $\eta_I$, $\eta_D$, and $\alpha$ need to be optimized for
the SE-ODIA. Fig. \ref{fig:rates_eta} shows the sum-rate performance
of the proposed SE-ODIA for varying $\eta_I$ or $\eta_D$ with two
different $\alpha$ values when $K=3$, $M=4$, $L=2$, $S=2$, and
$N=20$. To obtain the sum-rate according to $\eta_I$, $\eta_D$ was
fixed to $1$. Similarly, for the sum-rate according to $\eta_D$,
$\eta_I$ was fixed to $1$. If $\eta_I$ is too small, then there may
not be eligible users that satisfy the conditions $\mathsf{C}_1$ and
$\mathsf{C}_2$ in (\ref{eq:C}). Thus, \textit{scheduling
outage}~\footnote{It indicates the situation that there are no users
who are eligible for scheduling.} can occur frequently and the
achievable sum-rate becomes low. On the other hand, if $\eta_I$ is
too large, then the received interference at users may not be
sufficiently suppressed. Thus, the achievable sum-rate converges to
that of the system without interference suppression. Similarly, if
$\eta_D$ is too large, then the scheduling outage occurs; and if
$\eta_D$ is too small, then desired channel gains cannot be
improved. The orthogonality parameter $\alpha$ plays a similar role;
if $\alpha$ is too small, the cardinality of the user pool
$|\mathcal{N}_s|$ often becomes smaller than $S$, and scheduling
outage happens frequently. If $\alpha$ is too large, then the
orthogonality of the effective channel vectors of the selected users
is not taken into account for scheduling. In short, the parameters
$\eta_I$, $\eta_D$, and $\alpha$ need to be carefully chosen to
improve the performance of the proposed SE-ODIA. In subsequent
sum-rate simulations, proper sets of $\eta_I$, $\eta_D$, and
$\alpha$ were numerically found for various $N$ and SNR values and
applied to the SE-ODIA.
\begin{figure}[t]
\begin{center}
  \includegraphics[width=0.7\textwidth]{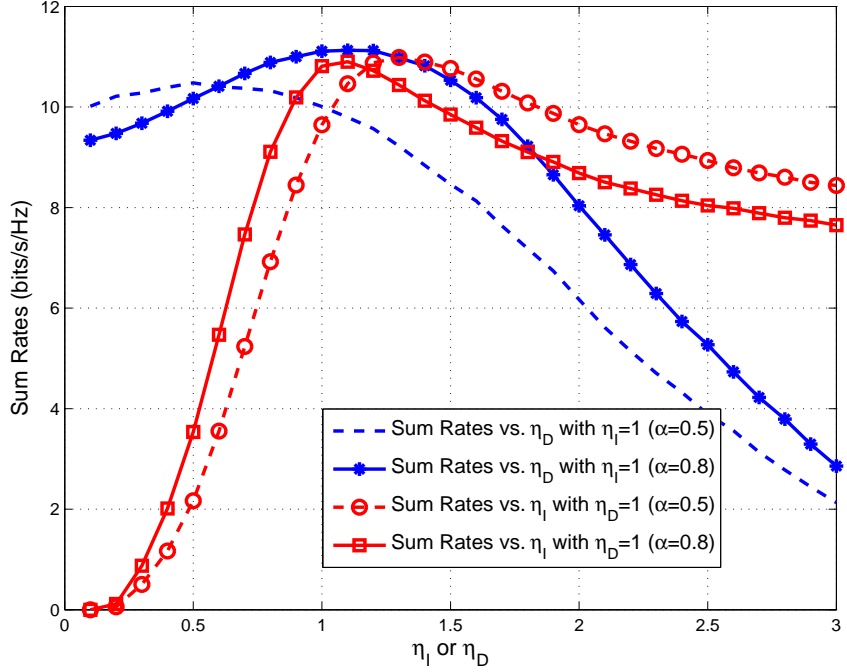}\\
  \caption{Sum-rates of SE-ODIA vs. $\eta_D$ or $\eta_I$ when $K=3$, $M=4$, $L=2$, $S=2$, and $N=20$.}\label{fig:rates_eta}
  \end{center}
\end{figure}
For instance, optimal $(\eta_I, \eta_D, \alpha)$ values that
maximize the sum-rate for a few cases are provided in Table
\ref{table:param}. It is seen that in the noise-limited low SNR
regime, large $\eta_D$ helps, whereas in the interference-limited
high SNR regime, small $\eta_I$ improves the sum-rate. On the other
hand, as $N$ increases, interference can be suppressed by choosing
smaller $\eta_I$ values.
\begin{table} \caption{Optimized parameters $(\eta_I, \eta_D, \alpha)$ for different SNRs and $N$ values}  \label{table:param}
\begin{center}
\begin{tabular}{|c|c|c|}
  \hline
   & $N$=20 & $N$=50 \\ \hline
  SNR=3dB & (2.5, 2.5, 0.8) & (2, 2.5, 0.8) \\ \hline
  SNR=21dB &  (1.5, 2, 0.8) & (1, 2, 0.8) \\
  \hline
\end{tabular}
\end{center}
\end{table}

Fig. \ref{fig:rates_SNR} shows the sum-rates for varying SNR values
when $K=3$, $M=4$, $L=2$, $S=2$, and (a) $N=20$ and (b) $N=50$.
 In the noise-limited low SNR regime, the sum-rate of the min-INR scheme is even lower than that of the max-SNR scheme, because $N$ is not large enough to suppress both intra- and inter-cell interference. The proposed ODIA outperforms the conventional schemes for SNRs larger than 2dB due to the combined effort of 1) transmit beamforming perfectly eliminating intra-cell interference and 2) receive beamforming effectively reducing inter-cell interference. The sum-rate performance of the ODIA with limited feedback~(ODIA-LF) improves as $n_f$ increases as expected. In practice, $n_f=6$ exhibits a good compromise between the number of feedback bits and sum-rate performance for the codebook dimension of 2 (i.e., $S=2$).
On the other hand, the proposed SE-ODIA achieves higher sum-rates
than the others including the ODIA for all SNR regime, because the
SE-ODIA improves desired channel gains and suppresses interference
simultaneously. Note however that the SE-ODIA includes the
optimization on the parameters for given SNR and $N$ and requires
the user scheduling method based on perfect CSI feedback, which
demands higher computational complexity than the user scheduling of
the ODIA.
As shown in Fig. \ref{fig:rates_SNR_N50}, the amount of sum-rates
improvement of the proposed ODIA schems for growing $N$ is much
larger than those of the conventional schemes.

\begin{figure}
\begin{center}
\subfigure[]{
\includegraphics[width=0.64\textwidth]{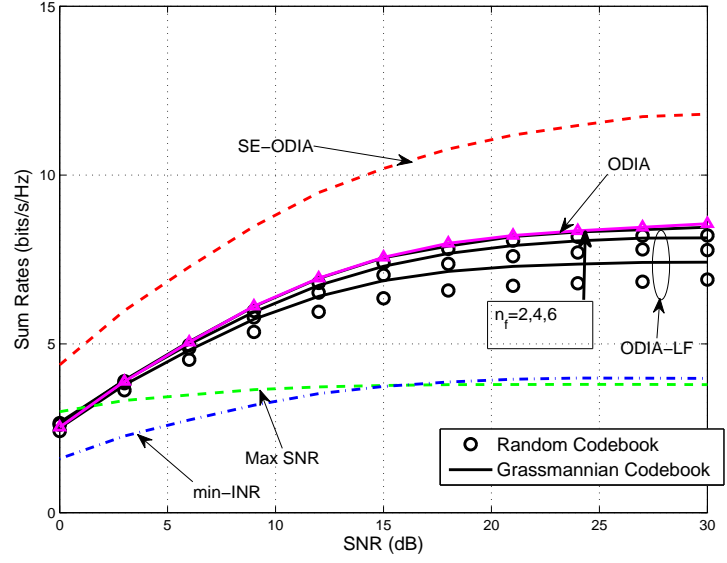}\label{fig:rates_SNR_N20} }
     \subfigure[]{
\includegraphics[width=0.64\textwidth]{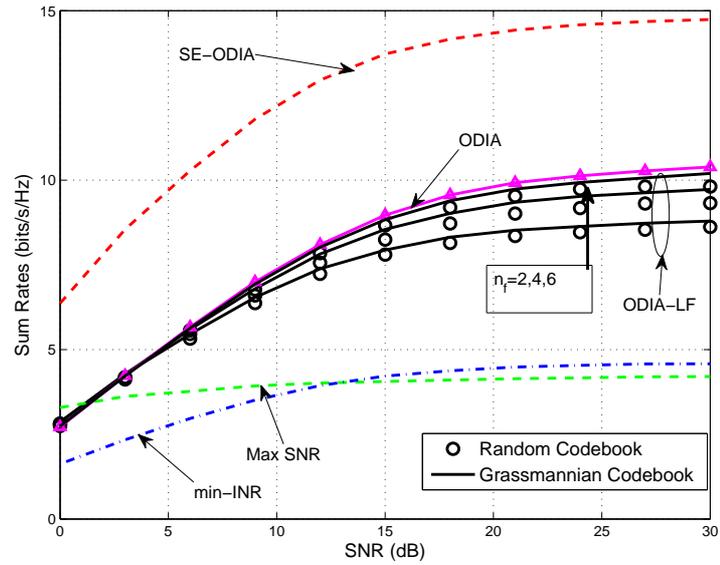}
  \label{fig:rates_SNR_N50}}
  \caption{Sum-rates versus SNR when $K=3$, $M=4$, $L=2$, $S=2$, and (a) $N=20$ (b) $N=50$.} \label{fig:rates_SNR}
\end{center}
\end{figure}

Fig. \ref{fig:rates_N_linear} shows the sum-rate performance of the
proposed ODIA schemes for varying number of users per cell, $N$,
when $K=3$, $M=4$, $L=2$, $S=2$, and SNR=$20$dB.
  For limited feedback, the Grassmannian codebook was employed. The sum-rates of the proposed ODIA schemes increase faster than the two conventional schemes, which implies that the user scaling conditions of the proposed ODIA schemes required for a given DoF or MUD gain are lowered than the conventional schemes, as shown in Theorems \ref{theorem:DoF} and \ref{theorem:MUD}.
\begin{figure}
\begin{center}
  \includegraphics[width=0.7\textwidth]{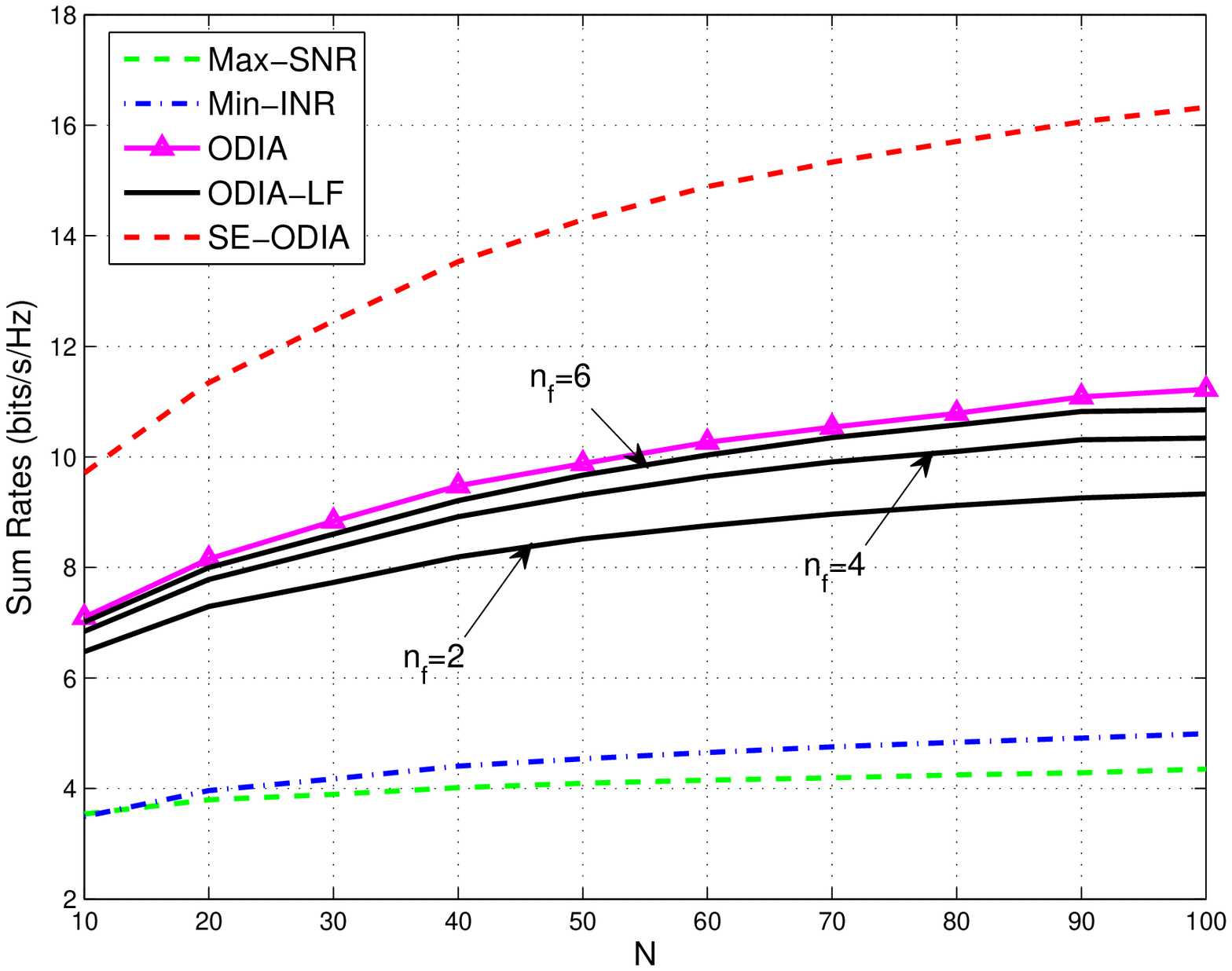}\\
  \caption{Sum-rates vs. $N$ when $K=3$, $M=4$, $L=2$, $S=2$, and SNR=20dB.}\label{fig:rates_N_linear}
  \end{center}
\end{figure}


\section{Conclusion} \label{SEC:Conc}
In this paper, we proposed an opportunistic downlink interference
alignment (ODIA) which intelligently combines user scheduling,
transmit beamforming, and receive beamforming for multi-cell
downlink networks. In the ODIA, the optimal DoF can be achieved with
more relaxed user scaling condition $N=\left(
\mathsf{SNR}^{(K-1)L-S+1}\right)$. To the best of our knowledge,
this user scaling condition is the best known to date. We also
considered a limited feedback approach for the ODIA, and analyzed
the minimum number of feedback bits required to achieve the same
user scaling condition of the ODIA with perfect feedback. We found
that both Grassmannian and random codebooks yield the same condition
on the number of required feedback bits. Finally, a spectrally
efficient ODIA (SE-ODIA) was proposed to further improve the
sum-rate of the ODIA, in which optimal multiuser diversity can be
achieved even in the presence of inter-cell interference. Through
numerical results, it was shown that the proposed ODIA schemes
significantly outperform the conventional interference management
schemes  in practical environments.

\appendices
\section{Proof of Lemma \ref{lemma:CDF_decay}} \label{app:interf_decay}
Recall that the selected users' $\eta^{[i,j]}$ are the minimum $S$
values out of $N$ i.i.d. random variables. For given $S$, suppose
the worse performance case where $N$ users are randomly divided into
$S$ subgroups with $N/S$ users per each and where one user with the
minimum $\eta^{[i,j]}$ is selected for each subgroup. Thus,
$\eta^{[i,j]}$ is the minimum of $N/S$ i.i.d. random variables. At
this point, let us define $\beta$ such that
\begin{align}\label{eq:beta_def}
\textrm{Pr}\left\{ \eta^{[i,j]} \le \frac{1}{\beta} \right\} =
\frac{S}{N}.
\end{align}
Note that since $\eta^{[i,j]}$ only decreases with respect to $N$,
$\beta$ also decreases as $N$ increases from (\ref{eq:beta_def}).
In addition, since the CDF of $\eta^{[i,j]}$ obtained from
(\ref{eq:LIF_beamforming_simple}) is the same as that of the
scheduling metric of the MIMO IMAC \cite{H_Yang13_TWC}, the CDF of
$\eta^{[i,j]}$ is given by \cite[Lemma 1]{H_Yang13_TWC}
\begin{equation} \label{eq:F_phi}
F_{\eta}(x) = c_0 x^{(K-1)S-L+1} + o\left(x^{(K-1)S-L+1} \right),
\end{equation}
where $c_0$ is a constant determined by $K$, $S$, and $L$. Thus, we
have $\textrm{Pr}\left\{ \eta^{[i,j]} \le \frac{1}{\beta} \right\} =
c_0 \beta^{-\tau} + o\left( \beta^{-\tau}\right)$ from Lemma
\ref{lemma:CDF_scaling}, where  $\tau = (K-1)S-L+1$, and hence
$\beta = \Theta \left( N^{1/\tau}\right)$.
In addition, since selected user's $1/\eta^{[i,j]}$ is the maximum
out of $N/S$ reversed scheduling metrics, it can be shown from
(\ref{eq:beta_def}) that
\begin{align}
\textrm{Pr}\left\{ \frac{1}{\eta^{[i,j]}|_{\textrm{selected}}} \le
\beta \right\} &= \left( 1- \frac{1}{N/S}\right)^{N/S}.
\end{align}
Therefore, the Markov inequality yields
\begin{align}
\label{eq:decay_last0}E\left\{ \frac{1}{\eta^{[i,j]}|_{\textrm{selected}}}  \right\}& \ge \beta \cdot \textrm{Pr} \left\{ \frac{1}{\eta^{[i,j]}|_{\textrm{selected}}} \ge \beta \right\} \\
& = \beta \cdot \left( 1- \left( 1- \frac{1}{N/S}\right)^{N/S}\right)\\
\label{eq:decay_last}& =  \Theta \left( N^{1/\tau}\right),
\end{align}
where (\ref{eq:decay_last}) follows from the fact that $\beta =
\Theta\left( N^{1/\tau}\right)$ and $\left( 1-
\frac{1}{N/S}\right)^{N/S}$ converges to a constant for increasing
$N$.

\section{Proof of Theorem \ref{th:codebook}} \label{app:th_codebook}

i) Grassmannian codebook \\
For the Grassmannian codebook, the chordal distance between any two
codewords is the same, i.e., $\sqrt{1-\left|
\mathbf{c}_i^{\mathsf{H}} \mathbf{c}_j\right|^2} = d,$, $\forall i
\neq j$. The Rankin, Gilbert-Varshamov, and Hamming bounds on the
chordal distance give us \cite{A_Barg02_TIT,J_Conway96_EM,
W_Dai08_TIT}
\begin{equation} \label{eq:dmin_bound}
{d^{[i,j]}}^2 \le  \min \left\{\frac{1}{2},
\frac{(S-1)N_f}{2S(N_f-1)}, \left(
\frac{1}{N_f}\right)^{1/(S-1)}\right\}.
\end{equation}
The bound in (\ref{eq:dmin_bound}) is reduced to the third bound as
$N_f$ increases, thus providing arbitrarily tight upper-bound on
${d^{[i,j]}}^2$. Thus,  the first term of (\ref{eq:Delta}) remains
constant if
\begin{align}
\left(d^{\max}_i\right)^2\delta_1\cdot \mathsf{SNR} & \le \left(
\frac{1}{N_f}\right)^{1/(S-1)}\delta_1\cdot \mathsf{SNR}\le
\epsilon'.
\end{align}
This is reduced to $N_f^{-1/(S-1)} \le \epsilon' \delta_1^{-1}
\mathsf{SNR}^{-1}$,
 or equivalently (\ref{eq:nf_cond0}).
Now, if (\ref{eq:nf_cond0}) holds true, $d_i^{\max}$ tends to be
arbitrarily small as SNR increases, and thus the second term of
(\ref{eq:Delta}) is dominated by the first term.
Therefore, if $n_f$ scales with respect to $\log_2(\mathsf{SNR})$ as
(\ref{eq:nf_cond0}), the residual intra-cell interference
$\Delta^{[i,j]}$ remains constant.

ii) Random codebook \\
In a random codebook, each codeword $\mathbf{c}_k$ is chosen
isotropically and independently from the $L$-dimensional hyper
sphere, and thus the maximum chordal distance of a random codebook
is unbounded. Since ${d^{[i,j]}}^2$ is the minimum of $N_f$ chordal
distances resulting from $N_f$ independent codewords, the CDF of
${d^{[i,j]}}^2$ is given by  \cite{C_AuYeung07_TWC,N_Jindal06_TIT}
\begin{equation} \label{eq:F_d_def}
F_d(z) \triangleq \textrm{Pr}\left\{{d^{[i,j]}}^2\le z\right\} =
1-\left(1-z^{S-1}\right)^{N_f}.
\end{equation}

From (\ref{eq:Delta}), the second term of (\ref{eq:P_def3}) can be
bounded by
\begin{align} \label{eq:PPP}
&\textrm{Pr} \Bigg\{\Delta^{[i,j]} \le \epsilon', \forall i\in \mathcal{K}, j\in \mathcal{S}  \Bigg\} \nonumber \\
&\hspace{10pt}\ge \textrm{Pr}\left\{\left(d^{\max}_i\right)^2\delta_1\cdot \mathsf{SNR}\le \epsilon', \forall i\in \mathcal{K}\right\} \nonumber \\
 & \hspace{30pt}\times \textrm{Pr} \left\{ \sum_{k=2}^{\infty}  \left(d^{\max}_i\right)^{2k}\delta_k \cdot \mathsf{SNR} \le \epsilon', \forall i\in \mathcal{K} \right\}.
\end{align}
Subsequently, we have
\begin{align} \label{eq:RV_P1}
&\textrm{Pr}\left\{\left(d^{\max}_i\right)^2\delta_1\cdot
\mathsf{SNR}\le \epsilon'\right\}  = \prod_{k=1}^{S}
\textrm{Pr}\left\{\left(d^{[k,i]}\right)^2\delta_1\cdot
\mathsf{SNR}\le \epsilon'\right\},
\end{align}
which follows from the fact that $d^{[k,i]}$ and $d^{[m,i]}$ are
independent for $k\neq m$. From (\ref{eq:F_d_def}) we have
\begin{align} \label{eq:RV_P1_2}
&\textrm{Pr}\left\{\left(d^{[k,i]}\right)^2\delta_1\cdot \mathsf{SNR}\le \epsilon'\right\}  \nonumber \\
&\hspace{20pt} =  1-\left(1-{\epsilon'}^{S-1}\delta_1^{-S+1}\left(
\mathsf{SNR}\right)^{-(S-1)}\right)^{N_f}.
\end{align}
Therefore, $\lim_{\mathsf{SNR}\rightarrow \infty}
\textrm{Pr}\left\{\left(d^{\max}_i\right)^2\delta_1\cdot
\mathsf{SNR}\le \epsilon'\right\}=1$  if and only if $N_f = \omega
\left( \mathsf{SNR}^{S-1}\right)$,
or equivalently (\ref{eq:nf_cond0}).
Now, if (\ref{eq:nf_cond0}) holds true, $d_{i}^{\max}$ tends to
arbitrarily small with high probability as SNR increases. Therefore,
the second term of (\ref{eq:Delta}) is dominated by the first term,
and hence $\textrm{Pr} \left\{\Delta^{[i,j]} \le \epsilon', \forall
i\in \mathcal{K}, j\in \mathcal{S}  \right\}$ in (\ref{eq:PPP})
tends to 1.

\section{Proof of Lemma \ref{lemma:N_card}} \label{app:Ns_cardinality}
Let us define the set $\Pi_s$ by
\begin{align}
\Pi_s \triangleq \left\{ \mathbf{h}\in \mathbb{C}^{S \times 1}:
\frac{{\mathbf{h}}^{\mathsf{H}} \mathbf{v}}{\| \mathbf{h}\|
\|\mathbf{v}\|}  <\alpha, \forall \mathbf{v}\in \textrm{span}\left(
\mathbf{b}_{1}^{[i]}, \ldots, \mathbf{b}_{s-1}^{[i]}\right)
\right\}.
\end{align}
Since the $s$-th user pool is determined only by checking the
orthogonality to the chosen users' channel vectors, for arbitrarily
large $N$, we have the followings by the law of large numbers:
\pagebreak[0]
\begin{align}
\label{eq:N_card1} |\mathcal{N}_s| &\approx N \cdot
\textrm{Pr}\left\{ \mathbf{h}\in \mathbb{C}^{S \times 1}:
\frac{{\mathbf{h}}^{\mathsf{H}} {\mathbf{b}_{s}^{[i]}}'}{\|
\mathbf{h}\| \|\mathbf{b}_{s'}^{[i]}\|}  <\alpha, s'=1, \ldots, s-1
\right\} \\ \pagebreak[0]
& \ge N \cdot \textrm{Pr}\left\{ \mathbf{h}\in \mathbb{C}^{S \times 1}: \mathbf{h} \in \Pi_s \right\} \\
&= N \cdot I_{\alpha^2}(s-1, S-s+1) \\ \pagebreak[0]
\label{eq:N_card2}&\ge N \cdot \alpha^{2(S-1)}, \pagebreak[0]
\end{align}
where $I_{\alpha^2}$ is the regularized incomplete beta function
(See \cite[Lemma 3]{T_Yoo06_JSAC}), and (\ref{eq:N_card2}) follows
from $I_{\alpha^2}(s-1, S-s+1) \ge I_{\alpha^2}(S-1, 1) =
\alpha^{2(S-1)}$.

\section{Proof of (\ref{eq:ps_conv})} \label{app:MUD}
Since $\left\|\mathbf{b}_{s}^{[i]}\right\|^2$ is a Chi-squared
random variable with $2(S-s+1)$ degrees-of-freedom, for $\eta_D>2$,
we have
\begin{align}
\textrm{Pr} \left\{ \mathsf{C}_2  \right\} 
& = 1-\frac{\gamma((S-s+1), \eta_D/2)}{\Gamma(S-s+1 )} = \frac{\Gamma((S-s+1), \eta_D/2)}{\Gamma(S-s+1 )}\\
& = \sum_{m=0}^{S-s} e^{-(\eta_D/2)} \frac{{(\eta_D/2)}^{m}}{m!} \\
& = \frac{e^{-(\eta_D/2)}\cdot {(\eta_D/2)}^{S-s}}{(S-s)!}\left(1+ O\left({(\eta_D/2)}^{-1}\right)\right)\\
\label{eq:C2_final}&\ge \frac{e^{-(\eta_D/2)}}{(S-s)!},
\end{align}
where $\Gamma(s,x) = \int_{x}^{\infty} t^{s-1}e^{-t}dt$ is the upper
incomplete gamma function and $\gamma(s,x)=
\int_{0}^{x}t^{s-1}e^{-t}dt$ is the lower incomplete gamma function.

Note that from the CDF of $\eta^{[i,j]}$ (See \cite[Lemma
1]{H_Yang13_TWC}), $\textrm{Pr} \left\{ \eta^{[i,j]} \le
\eta_I\right\} = c_0 \eta_I^{\tau} + o(\eta_I^{\tau})$, where $\tau
= (K-1)S-L+1$. Thus, from (\ref{eq:eta_D_choice}),
(\ref{eq:eta_I_choice}), and (\ref{eq:C2_final}),  (\ref{eq:p_s2})
can be bounded by
\begin{align} \label{eq:p_s_final}
\mathsf{p}_s & \gtrsim 1- \Bigg( 1- \left( c_0(\epsilon_I)^{\tau} {\mathsf{SNR}}^{-\tau} + \Omega\left( {\mathsf{SNR}}^{-(\tau-1)}\right)\right)\nonumber \\
& \hspace{100pt} \times \frac{N^{-(\epsilon_D/2)}}{(S-s)!}\Bigg)^{N
\cdot \alpha^{2(S-1)}}.
\end{align}
The right-hand side of (\ref{eq:p_s_final}) converges to 1 for
increasing SNR if and only if
\begin{align}\label{eq:scaling_MUD}
&\lim_{\mathsf{SNR}\rightarrow \infty} \left(N \cdot \alpha^{2(S-1)}\right) \cdot \left( c_0 (\epsilon_I)^{\tau} {\mathsf{SNR}}^{-\tau} + \Omega\left( {\mathsf{SNR}}^{-(\tau-1)}\right)\right)\nonumber \\
& \hspace{50pt}\times  \frac{N^{-(\epsilon_D/2)}}{(S-s)!}
\rightarrow \infty.
\end{align}
Since the left-hand side of (\ref{eq:scaling_MUD}) can be written by
$ \tilde{c}_0\frac{N^{1-(\epsilon_D/2)}}{\mathsf{SNR}^{\tau}} +
\tilde{c}_1\frac{N^{1-(\epsilon_D/2)}}{o\left(\mathsf{SNR}^{\tau}\right)}$,
where $\tilde{c}_0$ and $\tilde{c}_1$ are positive constants
independent of SNR and $N$, it tends to infinity for increasing SNR,
and thereby  $\mathsf{p}_s$ tends to 1 if and only if $N = \omega
\left( \mathsf{SNR}^{\frac{(K-1)S-L+1}{1-(\epsilon_D/2)}}\right)$.



\end{document}